%% file: dsfvs-arxiv.tex
\documentclass[onecolumn,11pt,letter]{article}
\usepackage{epsfig}
\usepackage{epstopdf}
\usepackage{graphicx}
\usepackage{amsfonts}
\usepackage{amsmath}
\usepackage{amsthm}
\usepackage{proof}
\usepackage{latexsym}
\usepackage{amssymb}
\usepackage{color}
\usepackage{comment}
\usepackage[ruled]{algorithm2e}
\usepackage[margin=1in]{geometry}
\usepackage[noend]{algorithmic}
\usepackage{tweaklist}
\usepackage{amsmath}
\usepackage{amssymb}
\usepackage{latexsym}
\usepackage{authblk}

\newtheorem{theorem}{Theorem}[section]
\newtheorem{lemma}[theorem]{Lemma}

\newtheorem{definition}[theorem]{Definition}

\newtheorem{remark}[theorem]{Remark}
\DeclareMathAlphabet{\mathcal}{OMS}{cmsy}{m}{n}
\usepackage{tweaklist}
\renewcommand{\enumhook}{\setlength{\itemsep}{-0.5ex}}
\renewcommand{\itemhook}{\setlength{\itemsep}{-0.5ex}}

\newcommand{\executeiffilenewer}[3]{%
\ifnum\pdfstrcmp{\pdffilemoddate{#1}}%
{\pdffilemoddate{#2}}>0%
{\immediate\write18{#3}}\fi%
}

\newcounter{note}[section]

\newcommand{\DFVS}{\textsc{Directed Feedback Vertex Set}}
\newcommand{\SDFVS}{\textsc{Subset Directed Feedback Vertex Set}}
\newcommand{\ESDFVS}{\textsc{Edge Subset Directed Feedback Vertex Set}}
\newcommand{\dfvs}{\textsc{DFVS}}
\newcommand{\sdfvs}{\textsc{Subset-DFVS}}
\newcommand{\esdfvs}{\textsc{Edge-Subset-DFVS}}

\newcommand{\branch}{\textsc{Branch}}
\newcommand{\torso}{\textup{\texttt{torso}}}
\newcommand{\randset}{\textup{RandomSet}}
\newcommand{\derandset}{\textup{DeterministicSets}}
\newcommand{\X}{{\mathcal{X}}}
\newcommand{\F}{{\mathcal{F}}}

\newcommand{\IS}{{\mathcal{I}_k}}
\begin{document}

\title{Directed Subset Feedback Vertex Set is Fixed-Parameter Tractable\thanks{A preliminary version of this paper appeared in ICALP 2012~\cite{DBLP:conf/icalp/ChitnisCHM12}.}}

\author[1]{Rajesh Chitnis\thanks{Supported in part by NSF CAREER award 1053605, ONR YIP award N000141110662, DARPA/AFRL award FA8650-11-1-7162, a
University of Maryland Research and Scholarship Award (RASA), a Graduate Student International Research Fellowship from the
University of Maryland, ERC Starting Grant PARAMTIGHT (No. 280152) and a Simons Award for Graduate Students in Theoretical Computer Science.}}
\author[2]{Marek Cygan\thanks{Supported in part by ERC Starting Grant NEWNET 279352, NCN grant N206567140 and Foundation for Polish Science.}}
\author[1]{MohammadTaghi Hajiaghayi\thanks{Supported in part by NSF CAREER award 1053605, ONR YIP award N000141110662, DARPA/AFRL award FA8650-11-1-7162, a University of Maryland Research and Scholarship Award (RASA)}}
\author[3]{D{\'a}niel Marx\thanks{Supported by ERC Starting Grant PARAMTIGHT (No. 280152) and OTKA grant NK105645.}}

\affil[1]{Department of Computer Science, University of Maryland at College Park, USA. Email:
\tt{\{rchitnis,hajiagha\}@cs.umd.edu}}

\affil[2]{Institute of Informatics, University of Warsaw, Poland. Email: \tt{cygan@mimuw.edu.pl}}

\affil[3]{Institute for Computer Science and Control, Hungarian Academy of Sciences (MTA SZTAKI), Budapest, Hungary. Email:
\tt{dmarx@cs.bme.hu}}

\renewcommand\Authands{ and }

%\author{
%Rajesh Chitnis
%  \thanks{Department of Computer Science, University of Maryland at
%  College Park, USA. Email: \tt{rchitnis@cs.umd.edu}. Supported in part by NSF CAREER award 1053605, ONR YIP award N000141110662, DARPA/AFRL award FA8650-11-1-7162, a
%  University of Maryland Research and Scholarship Award (RASA) and a Graduate Student International Research Fellowship from the University of Maryland}
%  \and
%Marek Cygan
%  \thanks{IDSIA, University of Lugano, Switzerland. Email: \tt{marek@idsia.ch}. Supported in part by ERC Starting Grant NEWNET 279352, NCN grant N206567140 and Foundation for Polish Science.} \and
%  MohammadTaghi Hajiaghayi\thanks{Department of Computer Science, University of Maryland at
%  College Park, USA and AT\&T Labs--Research. Email: \tt{hajiagha@cs.umd.edu}. Supported in part by NSF CAREER award 1053605, ONR YIP award N000141110662, DARPA/AFRL award FA8650-11-1-7162, a
%  University of Maryland Research and Scholarship Award (RASA)}
% \and D{\'a}niel Marx\thanks{Computer and Automation Research Institute, Hungarian Academy of Sciences
%(MTA SZTAKI), Budapest, Hungary, Email: \tt{dmarx@cs.bme.hu}. Supported by ERC Starting Grant PARAMTIGHT (No. 280152) }
%}

%\institute{ , \\ \email{\{rchitnis, hajiagha\}@cs.umd.edu} \and
%  Institute of Informatics, University of Warsaw, Poland,
%\email{} \and } }

\maketitle

\begin{abstract}
Given a graph $G$ and an integer $k$, the \textsc{Feedback Vertex Set} (\textsc{FVS}) problem asks if there is a vertex set $T$ of
size at most $k$ that hits all cycles in the graph. The first fixed-parameter algorithm for \textsc{FVS} in undirected graphs appeared in a monograph of Mehlhorn in 1984. The fixed-parameter
  tractability status of \textsc{FVS} in directed graphs was a long-standing open problem until Chen et al.~(STOC '08, JACM '08) showed that
  it is fixed-parameter tractable by giving a $4^{k}k!n^{O(1)}$ time
  algorithm. There are two subset versions of this problems: we are given an
  additional subset $S$ of vertices (resp., edges) and we want to hit
  all cycles passing through a vertex of $S$ (resp. an edge of
  $S$); the two variants are known to be equivalent in the parameterized sense. Recently, the \textsc{Subset
    Feedback Vertex Set} problem in undirected graphs was shown to be FPT by
  Cygan et al.~(ICALP '11, SIDMA '13) and independently by Kakimura et al.~(SODA
  '12). We generalize the result of Chen et al. (STOC '08, JACM '08) by showing
  that \textsc{Subset Feedback Vertex Set} in directed graphs can be
  solved in time $2^{O(k^3)}n^{O(1)}$, i.e., FPT parameterized by size
  $k$ of the solution. By our result, we complete the picture for
  feedback vertex set problems and their subset versions in undirected
  and directed graphs.

  The technique of random sampling of important separators was used by
  Marx and Razgon (STOC '11, SICOMP '14) to show that \textsc{Undirected Multicut}
  is FPT and was generalized by Chitnis et al. (SODA '12, SICOMP '13) to directed
  graphs to show that \textsc{Directed Multiway Cut} is FPT. Besides
  proving the fixed-parameter tractability of \textsc{Directed Subset
    Feedback Vertex Set}, we reformulate the random sampling of
  important separators technique in an abstract way that can be used
  for a general family of transversal problems.  We believe this
  general approach will be useful for showing the fixed-parameter
  tractability of other problems in directed graphs.  Moreover, we
  modify the probability distribution used in the technique to achieve
  better running time; in particular, this gives an improvement from $2^{2^{O(k)}}$ to $2^{O(k^2)}$ in the
  parameter dependence of the \textsc{Directed Multiway Cut} algorithm
  of Chitnis et al. (SODA '12, SICOMP '13).
% In this paper, we give a general family of problems (which includes \textsc{Directed
% Multiway Cut} and \textsc{Directed Subset Feedback Vertex Set} among others) for which the random sampling of important
% separators technique can be used to obtain a set which is disjoint from a minimum solution and covers its ``shadow.''
\end{abstract}

\input{intro}
\input{preliminaries}
\input{covering-shadow}
\input{reducing-instance}
\input{finding-shadowless}
\input{fpt-disjoint}
\input{conclusion}

\section*{Acknowledgements}

We thank Marcin Pilipczuk for pointing out a missing case in an earlier version of the algorithm.

\bibliographystyle{splncs03}
\bibliography{docsdb}

\end{document}

%% file: intro.tex
\section{Introduction}

The \textsc{Feedback Vertex Set} (\textsc{FVS}) problem has been one of the most extensively studied problems in the
parameterized complexity community. Given a graph $G$ and an integer $k$, it asks if there is a set $T\subseteq V(G)$ of
size at most $k$ which hits all cycles in $G$. The FVS problem in both undirected and directed graphs was shown to be NP-hard by
Karp~\cite{karp-np-hardness}. A generalization of the \textsc{FVS} problem is \textsc{Subset Feedback Vertex Set}
(\textsc{SFVS}): given a subset $S\subseteq V(G)$ (resp., $S\subseteq E(G)$), find a set $T\subseteq V(G)$ of size at most $k$
such that $T$ hits all cycles passing through a vertex of $S$ (resp., an edge of $S$). It is easy to see that $S=V(G)$ (resp.,
$S=E(G)$) gives the \textsc{FVS} problem.

As compared to undirected graphs, \textsc{FVS} behaves quite differently on directed graphs. In particular the trick of replacing
each edge of an undirected graph $G$ by arcs in both directions does not work: every feedback vertex set of the resulting
digraph is a vertex cover of $G$ and vice versa. Any other simple transformation does not seem possible either and thus the
directed and undirected versions are very different problems. This is reflected in the best known approximation ratio for the
directed versions as compared to the undirected problems: \textsc{FVS} in undirected graphs has an
2-approximation~\cite{bafna-2-approx-ufvs} while \textsc{FVS} in directed graphs has an $O(\log |V(G)|\log \log
|V(G)|)$-approximation~\cite{even-dfvs-approx,seymour-dfvs-approx}. The more general \textsc{SFVS} problem in undirected
graphs has an 8-approximation ~\cite{even-8-approx-sufvs} while the best-known approximation ratio in directed graphs is
$O(\min \{ \log |V(G)|\log \log |V(G)|,\\
 \log^{2} |S| \} )$~\cite{even-dfvs-approx}.

Rather than finding approximate solutions in polynomial time, one can look for exact solutions in time that is
superpolynomial, but still better than the running time obtained by brute force solutions. In both the directed and the
undirected versions of the feedback vertex set problems, brute force can be used to check in time $n^{O(k)}$ if a solution of
size at most $k$ exists: one can go through all sets of size at most $k$. Thus the problem can be solved in polynomial time if
the optimum is assumed to be small. In the undirected case, we can do significantly better: since the first FPT algorithm for
\textsc{FVS} in undirected graphs by Mehlhorn~\cite{mehlhorn} almost 30 years ago, there have been a number of
papers~\cite{ufvs-2,ufvs-1,ufvs-3,ufvs-ic,mikefellows-fvs,df-fvs,wernicke-fvs,ufvs-7,saket-fvs-1,saket-fvs-2} giving faster
algorithms and the current fastest (randomized) algorithm runs in time $O^*(3^k)$~\cite{cut-and-count} (the $O^{*}$ notation
hides all factors that are polynomial in the size of input). That is, undirected \textsc{FVS} is fixed-parameter tractable
parameterized by the size of the solution. Recall that a problem is \emph{fixed-parameter tractable} (FPT) with a particular
parameter $k$ if it can be solved in time $f(k)n^{O(1)}$, where $f$ is an arbitrary function depending only on $k$; see
\cite{downey-fellows,flum-grohe,niedermeier} for more background. For digraphs, the fixed-parameter tractability status of
\textsc{FVS} was a long-standing open problem (almost 16 years) until Chen et~al.~\cite{chen-dfvs} resolved it by giving an
$O^{*}(4^{k}k!)$ algorithm. This was recently generalized by Bonsma and Lokshtanov~\cite{fvs-mixed-graphs} who gave a
$O^{*}(47.5^{k}k!)$ algorithm for \textsc{FVS} in mixed graphs, i.e., graphs having both directed and undirected edges.

In the more general \textsc{Subset Feedback Vertex Set} problem, an additional subset $S$ of vertices is given and we want to
find a set $T\subseteq V(G)$ of size at most $k$ that hits all cycles passing through a vertex of $S$. In the edge version, we are
given a subset $S\subseteq E(G)$ and we want to hit all cycles passing through an edge of $S$. The vertex and edge versions
are indeed known to be equivalent in the parameterized sense in both undirected and directed graphs. Recently, Cygan et
al.~\cite{usfvs-icalp} and independently Kakimura et~al.~\cite{usfvs-ken} have shown that \textsc{Subset Feedback Vertex Set}
in undirected graphs is FPT parameterized by the size of the solution. Our main result is that \textsc{Subset Feedback Vertex
Set} in directed graphs is also fixed-parameter tractable parameterized by the size of the solution:

\begin{theorem}
{\bf (main result)} \textsc{Subset Feedback Vertex Set} (\sdfvs) in directed graphs can be solved in time $O^{*}(2^{O(k^3)})$.
\label{thm-main}
\end{theorem}

{\bf Our techniques.}  As a first step, we use the standard technique of {\em iterative compression} \cite{reed-smith-vetta-ic} to argue
that it is sufficient to solve the compression version of \sdfvs,
where we assume that a solution $T$ of size $k+1$ is given in the
input and we have to find a solution of size $k$.  Our algorithm for
the compression problem uses the technique of ``random sampling of
important separators,'' which was introduced by Marx and
Razgon~\cite{marx-razgon-stoc-11} for undirected \textsc{Multicut} and
generalized to directed graphs by Chitnis et
al.~\cite{directed-multiway-cut} to handle \textsc{Directed Multiway
  Cut}. We contribute two improvements to this technique on directed
graphs. First, we abstract out a framework that allows the clean and
immediate application of this technique for various problems. Second,
we modify the random selection process to improve the probability of
success. In particular, plugging in this improved result to the
\textsc{Directed Multiway Cut} algorithm of Chitnis et
al.~\cite{directed-multiway-cut}, the running time decreases from
$O^*(2^{2^{k}})$ to $O^{*}(2^{O(k^2)})$ thus giving an exponential improvement.
\begin{theorem}\label{th:dmc}
\textsc{Directed Multiway Cut} can be solved in time $O^*(2^{O(k^2)})$, where $k$ is the number of vertices to be deleted.
\end{theorem}

Our generic framework can be used for the following general family of
problems. Let $\F=\{F_1,F_2,\ldots,F_q\}$ be a set of subgraphs of a
graph $G$. An {\em $\F$-transversal} is a set of vertices that intersects
every $F_i$. We consider problems that can be formulated as
finding an $\F$-traversal. In particular, we will investigate
$\F$-transversal problems satisfying the following property: we say
that $\F$ is {\em $T$-connected} if for every $i\in [q]$, each vertex
of $F_i$ can reach some vertex of $T$ by a walk completely contained
in $F_i$ and is reachable from some vertex of $T$ by a walk completely
contained in $F_i$.
\begin{center}
\noindent\framebox{\begin{minipage}{6.00in}
\textbf{$\F$-transversal for $T$-connected $\F$}\\
\emph{Input}: A directed graph $G$, a positive integer $k$, and a set $T\subseteq V(G)$.\\
\emph{Parameter}: $k$\\
\emph{Question}: Does there exist an $\F$-transversal $W\subseteq V(G)$ with $|W|\leq k$, i.e., a set $W$ such that $F_i\cap
W\neq \emptyset$ for every $i\in [q]$?
\end{minipage}}
\end{center}
We emphasize here that the collection $\F$ is implicitly defined in a
problem specific-way and we {\em do not} assume that it is given
explicitly in the input, in fact, it is possible that $\F$ is
exponentially large. For example, in the \textsc{Directed Multiway
  Cut} problem we take $T$ as the set of terminals and $\F$ as the set
of all walks between different terminals; note that $\F$ is clearly
$T$-connected.  In the compression version of \sdfvs, we take $T$ as
the solution that we want to compress and $\F$ as the set of all
cycles containing a vertex of $S$; again, $\F$ is $T$-connected,
since if $T$ is a solution, then every cycle containing a vertex
of $S$ goes through $T$.

We define the ``shadow'' of a solution $X$ as those vertices that are
disconnected from $T$ (in either direction) after the removal of
$X$. A common idea in \cite{marx-razgon-stoc-11,directed-multiway-cut}
is to ensure first that there is a solution whose shadow is empty, as
finding such a shadowless solution can be a significantly easier task.
Our generic framework shows that for the $\F$-transversal problems
defined above, we can invoke the random sampling of important
separators technique and obtain a set which is disjoint from a minimum
solution and covers its shadow.  What we do with this set, however, is
problem specific. Typically, given such a set, we can use (some
problem-specific variant of) the ``torso operation'' to find an
equivalent instance that has a shadowless solution. Therefore, we can
focus on the simpler task of finding a shadowless solution; or more
precisely, finding any solution under the guarantee that a shadowless
solution exists. We believe our framework will provide a useful
opening step in the design of FPT algorithms for other transversal and
cut problems on directed graphs.

In the case of undirected \textsc{Multicut}~\cite{marx-razgon-stoc-11}, if there was a shadowless solution, then the problem
could be reduced to an FPT problem called \textsc{Almost 2SAT}~\cite{lp-daniel,almost-2-sat}. In the case of \textsc{Directed Multiway
Cut}~\cite{directed-multiway-cut}, if there was a solution whose shadow is empty, then the problem could be reduced to the
undirected version, which was known to be FPT~\cite{chen-improved-multiway-cut,DBLP:journals/toct/CyganPPW13,marx-2006}. For \sdfvs, the situation turns out to be a bit more complicated. As mentioned above, we
first use the technique of {iterative compression} to reduce the problem to an instance where we are given a solution $T$ and
we want to find a disjoint solution of size at most $k$. We define the ``shadows'' with respect to the solution $T$ that we
want to compress, whereas in~\cite{directed-multiway-cut}, the shadows were defined with respect to the terminal set $T$. The
``torso'' operation we define in this paper is specific to the \sdfvs\ problem, as it takes into account the set $S$ and modifies it accordingly. Furthermore, even after ensuring that there is a solution $T'$ whose shadow is empty, we are not done
unlike in~\cite{directed-multiway-cut}. We then analyze the structure of the graph $G\setminus T'$ and focus on the last strongly connected component in the topological ordering of this graph, i.e., the strongly connected component which can only have incoming edges from other strongly connected components.
We would like to find the subset of $T'$ that separates this component from the rest of the graph. In most cases, a pushing argument can be used to argue that this subset of $T'$ is an important separator, and hence we can branch on removing an important separator from the graph. However, due to the way the set $S$ interacts with the solution $T'$, there is a small number of vertices that behave in a special way. We need surprisingly complex arguments to handle these special vertices.

%On the other hand, the fact that we are dealing with a directed graph makes the problem significantly harder (recall that
%\textsc{Directed
%  Multicut} is W[1]-hard, thus it is expected that not every
%undirected argument generalizes to the directed case). After defining a proper notion of directed important separators, the
%non-trivial interaction amongst two kinds of ``shadows" forces us to do the random sampling of important separators in two
%independent steps.  In \cite{marx-stoc-2011}, the basic version of random sampling gives a running time
% that is double exponential in $p$; a more complicated sampling process allowed to bring down the running time from
%$O^{*}(2^{2^{O(p)}})$ to $O^{*}(2^{O(p^3)})$.  Directed graphs have a notion of weak versus strong connectivity and this
%difference does not allow us to extend the more complicated version of sampling to directed graphs. Therefore, it remains an
%open question if single-exponential running time can be achieved for \textsc{Directed Multiway Cut}.

The paper is organized as follows. Section~\ref{sec:preliminaries}
introduces notation and the preliminary steps of the algorithm,
including iterative compression. Section~\ref{sec:coveringshadow}
presents the general result on covering shadows of
$\F$-transversals. The remaining sections are specific to the
\textsc{Subset-DFVS} problem: they discuss how to use the techniques of
Section~\ref{sec:coveringshadow} to reduce the problem to instances
where the existence of shadowless solutions is guaranteed
(Section~\ref{sec:torso}) and how to find a solution under the
guarantee that a shadowless solution exists
(Section~\ref{sec:branch-imp-sep}); the full algorithm is summarized
in Section~\ref{sec:fpt-algorithm}. Finally Section~\ref{sec:concl-open-probl}
concludes the paper.
%%% Local Variables:
%%% mode: latex
%%% TeX-master: "dsfvs-arxiv"
%%% End:

%% file: preliminaries.tex
\section{Preliminaries}
\label{sec:preliminaries}

Observe that a directed graph contains no cycles if and only if it
contains no closed walks; moreover, there is a cycle going through $S$
if and only there is a closed walk going through $S$. For this reason,
throughout the paper we use the term closed walks instead of cycles, since it is
sometimes easier to show the existence of a closed walk and avoid
discussion whether it is a simple cycle or not. A feedback vertex set
is a set of vertices that hits all the closed walks of the graph.

\begin{definition}{\bf (feedback vertex set)}
\label{defn-directed-multiway-cut} Let $G$ be a directed graph. A set $T\subseteq V(G)$ is a \emph{feedback vertex set} of $G$
if $G\setminus T$ does not contain any closed walks.
\end{definition}

This gives rise to the \DFVS\ (\dfvs) problem where we are given a directed graph $G$ and we want to find if $G$ has a
feedback vertex set of size at most $k$. The \dfvs\ problem was shown to be FPT by Chen et al.~\cite{chen-dfvs}, answering a long-standing
open problem in the parameterized complexity community.
%\daniel{Do we really need the formal definition of DFVS?}

%\begin{center}
%\noindent\framebox{\begin{minipage}{4.50in}
%\textbf{\DFVS\ (\dfvs)}\\
%\emph{Input }: A directed graph $G=(V,E)$ and a positive integer $k$.\\
%\emph{Parameter }: $k$\\
%\emph{Question} : Does there exist a set $T\subseteq V(G)$ with $|T|\leq k$ such that $G\setminus T$ has no closed-walks?
%\end{minipage}}
%\end{center}

In this paper, we consider a generalization of the \dfvs\ problem where given a set $S\subseteq V(G)$, we ask if there exists a
vertex set of size $\leq k$ that hits all closed walks passing through $S$.

\begin{center}
\noindent\framebox{\begin{minipage}{6.00in}
\textbf{\SDFVS\ (\sdfvs)}\\
\emph{Input}: A directed graph $G$, a set $S\subseteq V(G)$, and a positive integer $k$.\\
\emph{Parameter}: $k$\\
\emph{Question}: Does there exist a set $T\subseteq V(G)$ with $|T|\leq k$ such that $G\setminus T$ has no closed walk
containing a vertex of S?
\end{minipage}}
\end{center}

It is easy to see that \dfvs\ is a special case of \sdfvs\ obtained by setting $S=V(G)$. We also define a variant of \sdfvs\,
where the set $S$ is a subset of edges. In this variant, we have destroy the following type of closed walks:

\begin{definition}
(\textbf{$S$-closed-walk}) Let $G$ be a directed graph and $S\subseteq E(G)$. A closed walk (starting and ending at same
vertex) $C$ in $G$ is said to be a \emph{$S$-closed-walk} if it contains an edge from $S$.
\end{definition}

\begin{center}
\noindent\framebox{\begin{minipage}{6.00in}
\textbf{\ESDFVS\ (\esdfvs)}\\
\emph{Input }: A directed graph $G$, a set $S\subseteq E(G)$, and a positive integer $k$.\\
\emph{Parameter }: $k$\\
\emph{Question} : Does there exist a set $T\subseteq V(G)$ with $|T|\leq k$ such that $G\setminus T$ has no $S$-closed-walks?
\end{minipage}}
\end{center}

The above two problems can be shown to be equivalent as follows.  If
$(G,S,k)$ is an instance of \sdfvs\, we create an instance $(G,S',k)$
of \esdfvs\ by taking $S'$ as the set of edges incident to any vertex
of $S$. Then any closed walk passing through a vertex of $S$ must pass
through an edge of $S'$, and conversely any closed walk passing
through an edge of $S'$ must contain a vertex from $S$.

On the other hand, given an instance $(G,S',k)$ of \esdfvs\, we create
an instance $(G',S,k)$ of \sdfvs\ where $G'$ is obtained from $G$ by
the following modification: For every directed edge $(u,v)\in E(G)$ we add a
new vertex $x_{uv}$ and path $u\rightarrow x_{uv}\rightarrow v$ of
length 2. We set $S=\{x_{e}\ :\ e\in S'\}$. Then any closed walk in
$G$ passing through an edge of $S'$ corresponds to a closed-walk in
$G'$ which must pass through a vertex of $S$, and conversely any
closed walk in $G'$ passing through a vertex of $S$ can be easily
converted to a closed walk in $G$ passing through an edge of
$S'$. Both the reductions work in polynomial time and do not change
the parameter. Therefore, in the rest of the paper we concentrate on
solving the \ESDFVS\ problem and we shall refer to both the above
problems as \sdfvs.

\subsection{Iterative Compression}

The first step of our algorithm is to use the technique of
\emph{iterative compression} introduced by Reed et
al.~\cite{reed-smith-vetta-ic}. It has been used to obtain faster FPT
algorithms for various
problems~\cite{ufvs-ic,chen-dfvs,mikefellows-fvs,saket-ic,wernicke-fvs,huffner-ic,marx-razgon-stoc-11,almost-2-sat}. We
transform the \sdfvs\ problem into the following problem:
\begin{center}
\noindent\framebox{\begin{minipage}{6.00in}
\textbf{\sdfvs\ \textsc{Compression}}\\
\emph{Input}: A directed graph $G$, a set $S\subseteq E(G)$, a positive integer $k$, and a set $T\subseteq V$
such that $G\setminus T$ has no $S$-closed-walks.\\
\emph{Parameter}: $k+|T|$\\
\emph{Question}: Does there exist a set $T'\subseteq V(G)$ with $|T'|\leq k$ such that $G\setminus T'$ has no
$S$-closed-walks?
\end{minipage}}
\end{center}

\begin{lemma}%$[\star]$\footnote{The proofs of the results labeled with $\star$ have been deferred to the full version of the paper.}
(\textbf{power of iterative compression}) \sdfvs\ can be solved by $O(n)$ calls to an algorithm for the \sdfvs\
\textsc{Compression} problem with $|T|\le k+1$. \label{lem:ic}
\end{lemma}
\begin{proof}
Let $V(G)=\{v_1,\ldots,v_n\}$ and for $i\in [n]$ let $V_i = \{v_1, \ldots v_i\}$. We construct a sequence of subsets $X_i
\subseteq V_i$, such that $X_i$ is a solution for $G[V_i]$. Clearly, $X_1=\emptyset$ is a solution for $G[V_1]$. Observe that
if $X_i$ is a solution for $G[V_i]$, then $X_i \cup \{v_{i+1}\}$ is a solution for $G[V_{i+1}]$. Therefore, for each $i\in
[n-1]$, we set $T = X_i \cup \{v_{i+1}\}$ and use, as a blackbox, an algorithm for \sdfvs\ \textsc{Compression}, to construct a
set $X_{i+1}$ that is a solution of size at most $k$ for $G[V_{i+1}]$. Note that if there is no solution for $G[V_i]$ for
some $i\in [n]$, then there is no solution for the whole graph $G$ and moreover, since $V_n = V(G)$, if all the calls to the
reduction problem are successful, then $X_n$ is a solution for the graph $G$.
\end{proof}

Now we transform the \sdfvs\ \textsc{Compression} problem into the following problem whose only difference is that the subset
feedback vertex set in the output must be disjoint from the one in the input:

\begin{center}
\noindent\framebox{\begin{minipage}{6.00in}
\textbf{\textsc{Disjoint \sdfvs\ Compression}}\\
\emph{Input}: A directed graph $G$, a set $S\subseteq E(G)$, a positive integer $k$, and a set $T\subseteq V$
such that $G\setminus T$ has no $S$-closed-walks.\\
\emph{Parameter}: $k+|T|$\\
\emph{Question}: Does there exist a set $T'\subseteq V(G)$ with $|T'|\leq k$ such that $T\cap T' = \emptyset$ and $G\setminus
T'$ has no $S$-closed-walks?
\end{minipage}}
\end{center}

\begin{lemma}%$[\star]$
(\textbf{adding disjointness}) \textsc{\sdfvs\ Compression} can be solved by $O(2^{|T|})$ calls to an algorithm for the
\textsc{Disjoint \sdfvs\ Compression} problem. \label{lem:disjoint}
\end{lemma}
\begin{proof}
  Given an instance $I=(G,S,T,k)$ of \sdfvs\ \textsc{Compression} we
  guess the intersection $X$ of $T$ and the subset feedback vertex set
  $T'$ in the output. We have at most $2^{|T|}$ choices for $X$. Then
  for each guess for $X$, we solve the \textsc{Disjoint \sdfvs\
    Compression} problem for the instance $I_X=(G\setminus
  X,S,T\setminus X,k-|X|)$. It is easy to see that if $T'$ is a
  solution for instance $I$ of \textsc{\sdfvs\ Compression}, then
  $T'\setminus X$ is a solution of instance $I_X$ of \textsc{Disjoint
    \sdfvs\ Compression} for $X=T'\cap T$. Conversely, if $T''$ is a
  solution to some instance $I_X$, then $T''\cup X$ is a solution for
  $X$.
\end{proof}

From Lemmas~\ref{lem:ic} and \ref{lem:disjoint}, an FPT algorithm for \textsc{Disjoint \sdfvs\ Compression} translates into an
FPT algorithm for \sdfvs\ with an additional blowup factor of $O(2^{|T|}n)$ in the running time.
%We formalize this in the next observation:
%\daniel{I don't think we need to state this observation}
%\begin{observation}
%\label{observation}
%An FPT algorithm for \textsc{Disjoint \sdfvs\ Compression} implies an FPT algorithm for \sdfvs.
%\end{observation}

%%% Local Variables:
%%% mode: latex
%%% TeX-master: "dsfvs-arxiv"
%%% End:

%% file: covering-shadow.tex
\section{General $\F$-transversal Problems: Covering the Shadow of a Solution}
\label{sec:coveringshadow}
The purpose of this section is to present the ``random sampling of important separators'' technique developed in
\cite{directed-multiway-cut} for \textsc{Directed Multiway Cut} in a generalized way that applies to \sdfvs\ as well. The
technique consists of two steps:
%The two main steps in the FPT algorithm for \textsc{Directed Multiway Cut}~\cite{directed-multiway-cut} were:

\begin{enumerate}
\item First find a set $Z$ \emph{small} enough to be disjoint from a solution $X$ (of size $\leq k$) but \emph{large}
    enough to cover the ``shadow" of $X$.
\item Then define a ``torso" operation that uses the set $Z$ to reduce the problem instance in such a way that $X$
    becomes a shadowless solution of the reduced instance.
\end{enumerate}

In this section, we define a general family of problems for which Step 1 can be efficiently performed. The general technique
to execute Step 1 is very similar to what was done for \textsc{Directed Multiway Cut}~\cite{directed-multiway-cut}. In
Section~\ref{sec:torso}, we show how Step 2 can be done for the specific problem of \textsc{Disjoint
  \sdfvs\ Compression}. First we start by defining separators and shadows. Following \cite{directed-multiway-cut}, we define separators in a generalized setting where we assume that the graph $G$ is equipped with a subset $V^\infty(G)$ of undeletable vertices and separators by definition have to be disjoint from this set.
This extension will be very convenient in the proofs of Section~\ref{sec:analysis-algorithm}.

\begin{definition}{\bf (separator)}\label{defn-sep}
Let $G$ be a directed graph and $V^{\infty}(G)$ be the set of distinguished (``undeletable'') vertices. Given two disjoint
non-empty sets $X,Y\subseteq V$,  we call a set $W\subseteq V\setminus (X\cup Y\cup V^{\infty})$ an \emph{$X-Y$ separator} if
there is no path from $X$ to $Y$ in $G\setminus W$. A set $W$ is a {\em minimal $X-Y$ separator} if no proper subset of $W$ is
an $X-Y$ separator.
\end{definition}

Note that here we explicitly define the $X-Y$ separator $W$ to be disjoint from $X$ and $Y$.
\begin{definition}{\bf (shadows)}
\label{defn-f-and-r-and-shadow} Let $G$ be graph and $T$ be a set of terminals. Let $W\subseteq V(G)\setminus V^{\infty}(G)$
be a subset of vertices.
\begin{enumerate}
\item The \emph{forward shadow} $f_{G,T}(W)$ of $W$ (with respect to $T$) is the set of vertices $v$ such that $W$ is a
    $T-\{v\}$ separator in $G$.
\item The \emph{reverse shadow} $r_{G,T}(W)$ of $W$ (with respect to $T$) is the set of vertices $v$ such that $W$ is a
    $\{v\}-T$ separator in $G$.
\end{enumerate}
The \emph{shadow} of $W$ (with respect to $T$) is the union of $f_{G,T}(W)$ and $r_{G,T}(W)$.
\end{definition}
That is, we can imagine $T$ as a light source with light spreading on the directed edges. The forward shadow is the set of
vertices that remain dark if the set $W$ blocks the light, hiding $v$ from $T$'s sight. In the reverse shadow, we imagine that
light is spreading on the edges backwards. We abuse the notation slightly and write $v-T$ separator instead of $\{v\}-T$
separator. We also drop $G$ and $T$ from the subscript if they are clear from the context. Note that $W$ itself is not in the
shadow of $W$ (as, by definition, a $T-v$ or $v-T$ separator needs to be disjoint from $T$ and $v$), that is, $W$ and $f_{G,T}(W)\cup
r_{G,T}(W)$ are disjoint. See Figure~\ref{fig:shadow} for an illustration.

\begin{figure}[t]
\centering
\def\svgwidth{0.4\linewidth}%
\executeiffilenewer{mway1.svg}{mway1.pdf}%
{inkscape -z -D --file=mway1.svg %
--export-pdf=mway1.pdf --export-latex}%
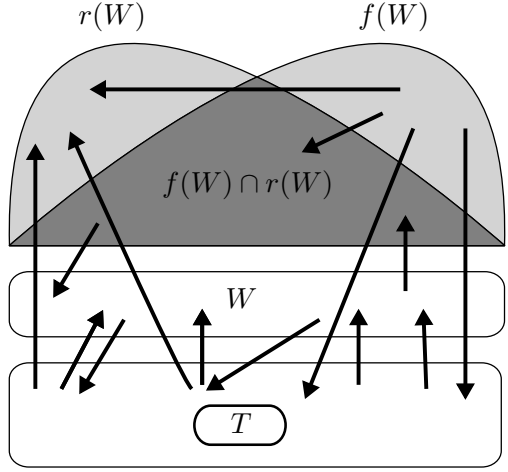%

\caption{For every vertex $v\in f(W)$, the set $W$ is a $T-v$
  separator. For every vertex $w\in r(W)$, the set $W$ is a $w-T$
  separator. For every vertex $y\in f(W)\cap r(W)$, the set $W$ is
  both a $T-y$ and $y-T$ separator. Finally for every $z\in
  V(G)\setminus [W\cup r(W)\cup f(W)\cup T]$, there are both $z-T$ and
  $T-z$ paths in the graph $G\setminus W$\label{fig:shadow}.}
%  Note that   every such vertex $z$ belongs to a strongly connected component of   $G\setminus W$ containing $T$ and there are no edges between these   components.}
\end{figure}

%\subsection{Covering the Shadow}

%In this section we present a general family of problems for which we can achieve Step 1.

Let $G$ be a directed graph and $T\subseteq V(G)$. Let
$\mathcal{F}=\{F_1,F_2,\ldots,F_q\}$ be a set of subgraphs of $G$. We
define the following property:

\begin{definition}
{\bf (T-connected)} Let $\mathcal{F}=\{F_1,F_2,\ldots,F_q\}$ be a set of subgraphs of $G$. For a set $T\subseteq V$, we say
that $\mathcal{F}$ is \emph{$T$-connected} if for every $i\in [q]\ $, each vertex of the subgraph $F_i$ can reach some vertex
of $T$ by a walk completely contained in $F_i$ and is reachable from some vertex of $T$ by a walk completely contained in
$F_i$. \label{defn:T-connected}
\end{definition}
For a set $\mathcal{F}$ of subgraphs of $G$, an {\em $\F$-transversal} is a set of vertices that intersects each subgraph in
$\mathcal{F}$.
%We note that the subgraphs in $\F$ are given implicitly to us.

\begin{definition}
{\bf ($\F$-transversal)} Let $\mathcal{F}=\{F_1,F_2,\ldots,F_q\}$ be a set of subgraphs of $G$. Then $W\subseteq V(G)$ is said to be an
\emph{$\F$-transversal} if $\forall\ i\in [q]$ we have $F_i\cap W\neq \emptyset$. \label{defn:F-transversal}
\end{definition}

The main result of this section is a randomized algorithm for producing a set that covers the shadow of some $\F$-transversal:
%\daniel{It would be sufficient to state only the derandomized version.}

\begin{theorem}{\bf (randomized covering of the shadow)} Let $T\subseteq V(G)$. There is an algorithm $\randset(G,T,k)$ that runs in $O^{*}(4^k)$ time
and returns a set $Z\subseteq V(G)$ such that for any  set $\F$ of $T$-connected subgraphs, if there exists an
$\F$-transversal of size $\leq k$, then the following holds with probability $2^{-O(k^2)}$: there is an $\F$-transversal $X$
of size $\leq k$ such that
\begin{enumerate}
\item $X\cap Z=\emptyset$ and
\item $Z$ covers the shadow of $X$.
\end{enumerate}
\label{thm:main-covering-shadow}
\end{theorem}

Note that $\F$ is {\em not} an input of the algorithm described by
Theorem~\ref{thm:main-covering-shadow}: the set $Z$ constructed in the
above theorem works for {\em every} $T$-connected set $\F$ of
subgraphs. Therefore, issues related to the representation of $\F$ do
not arise. Using the theory of splitters, we also prove the following derandomized version of
Theorem~\ref{thm:main-covering-shadow}:

\begin{theorem}{\bf (deterministic covering of the shadow)}
\label{thm:main-covering-shadow-derandomized}
Let $T\subseteq V(G)$.
We can construct a set $\{Z_1,Z_2,\ldots,Z_t\}$ with $t=2^{O(k^2)}\log^2 n$ in time $O^{*}(2^{O(k^2)})$ such that for any set $\F$ of $T$-connected subgraphs, if there exists an $\F$-transversal of size $\leq k$, then
there is an $\F$-transversal $X$ of size $\leq k$ such that for at least one $i\in [t]$ we have
\begin{enumerate}
\item $X\cap Z_i=\emptyset$ and
\item $Z_i$ covers the shadow of $X$.
\end{enumerate}
\end{theorem}

Sections~\ref{subsection-imp-separators}--\ref{sec:analysis-algorithm} are devoted to the proofs of Theorems~\ref{thm:main-covering-shadow}--\ref{thm:main-covering-shadow-derandomized}.

In the \textsc{Directed Multiway Cut} algorithm of Chitnis et
al.~\cite{directed-multiway-cut}, the set $T$ was the set of terminals
and the set $\F$ was the set of all walks from one vertex of $T$ to
another vertex of $T$. Clearly, $\F$ is $T$-connected: every vertex on
a walk from $T$ to $T$ satisfies the reachability conditions.  With
this interpretation,
Theorem~\ref{thm:main-covering-shadow-derandomized} generalizes
Theorem 4.11 of \cite{directed-multiway-cut} with a better running
time. Plugging Theorem~\ref{thm:main-covering-shadow-derandomized}
into the \textsc{Directed Multiway Cut} algorithm of
\cite{directed-multiway-cut} gives an $O^*(2^{O(k^2)})$ time
algorithm, proving Theorem~\ref{th:dmc}.

In \sdfvs, the set $T$ is the solution that we want to compress and
$\F$ is the set of all $S$-closed-walks passing through some vertex of
$T$. Again, $\F$ is $T$-connected: every $S$-closed-walk goes through
$T$ (as $T$ is a solution), hence any vertex on an $S$-closed-walk is
reachable from $T$, and some vertex of $T$ is reachable from every
vertex of the $S$-closed-walk.

%In the rest of this section, we
%will always assume the set $\F$ of subgraphs to be $T$-connected.

We say that an $\mathcal{F}$-transversal $T'$ is \emph{shadowless} if
$f(T')\cup r(T')=\emptyset$. Note that if $T'$ is a shadowless
solution, then each vertex of $G\setminus T'$ is reachable
from some vertex of $T$ and can reach some vertex of $T$.  In
Section~\ref{sec:torso}, we show that given an instance of
\textsc{Disjoint \sdfvs\ Compression} and a set $Z$ as in
Theorem~\ref{thm:main-covering-shadow}, we are able to transform the
instance using the torso operation in a way that guarantees the
existence of the shadowless solution for the reduced instance.  In
Section~\ref{sec:branch-imp-sep}, we will see how we can make progress
in \textsc{Disjoint \sdfvs\ Compression} if there exists a shadowless
solution: we identify a bounded-size set of vertices such that every
shadowless solution contains at least one vertex of this
set. Therefore, we can branch on including one vertex of this set into
the solution.

\subsection{Important separators and random sampling}
\label{subsection-imp-separators}

This subsection reviews the notion of important separators and the
random sampling technique introduced in
\cite{marx-razgon-stoc-11}. These ideas were later adapted and generalized for
directed graphs in~\cite{directed-multiway-cut}. We closely
follow~\cite{directed-multiway-cut}, but we deviate from it in two
ways: we state the results in the framework of $\F$-transversal
problems and improve the random selection and its analysis to achieve
better running time. Unfortunately, this means that we have to go step-by-step
through most of the corresponding arguments of
\cite{directed-multiway-cut}. While some of the statements and proofs
are almost the same as in \cite{directed-multiway-cut}, we give a
self-contained presentation without relying on earlier work (with the
exception of the proof of Lemma~\ref{number-of-imp-sep}).

\subsubsection{Important separators}

Marx~\cite{marx-2006} introduced the concept of \emph{important
  separators} to deal with the \textsc{Undirected Multiway Cut}
problem. Since then it has been used implicitly or explicitly in
\cite{chen-improved-multiway-cut,chen-dfvs,directed-multiway-cut,DBLP:conf/icalp/KratschPPW12,lokshtanov-marx-clustering,DBLP:conf/icalp/LokshtanovR12,marx-razgon-stoc-11,almost-2-sat}
in the design of fixed-parameter algorithms. In this section, we
define and use this concept in the setting of directed graphs. Roughly
speaking, an important separator is a separator of small size that is
\emph{maximal} with respect to the set of vertices on one side. Recall
that, as in Definition~\ref{defn-sep}, the graph $G$ has a set
$V^{\infty}(G)$ of undeletable vertices and an $X-Y$ separator is
defined to be disjoint from $X\cup Y \cup V^{\infty}(G)$.

\begin{definition}{\bf (important separator)}
\label{defn-imp-sep} Let $G$ be a directed graph and let $X,Y\subseteq V$ be two disjoint non-empty sets.  A minimal $X-Y$
separator $W$ is called an \emph{important $X-Y$ separator} if there is no $X-Y$ separator $W'$ with $|W'|\leq |W|$ and
$R^{+}_{G\setminus W}(X)\subset R^{+}_{G\setminus W'}(X)$, where $R^{+}_{A}(X)$ is the set of vertices reachable from $X$ in the graph
$A$.
\end{definition}

Let $X, Y$ be disjoint sets of vertices of an \emph{undirected graph}. Then for every $k\geq 0$, it is
known~\cite{chen-improved-multiway-cut,marx-2006} that there are at most $4^{k}$ important $X-Y$ separators of size at most
$k$ for any sets $X,Y$. The next lemma shows that the same bound holds for important separators even in directed graphs.

\begin{lemma}[\cite{directed-multiway-cut}]
{\bf (number of important separators)}
%appendix.}
\label{number-of-imp-sep} Let $X,Y\subseteq V(G)$ be disjoint sets in a directed graph $G$. Then for every
$k\geq 0$ there are at most $4^k$ important $X-Y$ separators of size at most $k$. Furthermore, we can enumerate all these
separators in time $O(4^k\cdot k(|V(G)+|E(G)|))$.
\end{lemma}
For ease of notation, we now define the
following collection of important separators:
\begin{definition}\label{defn-forward-reverse-impsep} Given a graph $G$, a set $T\subseteq V(G)$, and an integer $k$, the set $\IS$ contains the set $W\subseteq V(G)$ if $W$ is an important $v-T$ separator of size at most $k$ in $G$ for some
vertex $v$ in $V(G)\setminus T$.
\end{definition}

\begin{remark}
\label{remark:number-of-impseps} \emph{It follows from Lemma~\ref{number-of-imp-sep} that $|\IS|\le 4^k\cdot |V(G)|$ and we can enumerate the sets in $\IS$ in time $O^*(4^k)$.}
\end{remark}

We now define a special type of shadows which we use later for the random sampling:

\begin{definition}{\bf (exact shadows)}
\label{defn-exact-shadow} Let $G$ be a directed graph and $T\subseteq V(G)$ a set of terminals. Let $W\subseteq V(G)\setminus
V^{\infty}(G)$ be a subset of vertices. Then for $v\in V(G)$ we say that
\begin{enumerate}
\item $v$ is in the ``\emph{exact forward shadow}''  of $W$ (with respect to $T$) if $W$ is a minimal $T-v$ separator in $G$, and
\item $v$ is in the ``\emph{exact reverse shadow}''  of $W$ (with respect to $T$) if $W$ is a minimal $v-T$ separator in $G$.
\end{enumerate}
\end{definition}

\begin{figure}[t]
\centering
\includegraphics[width=6in]{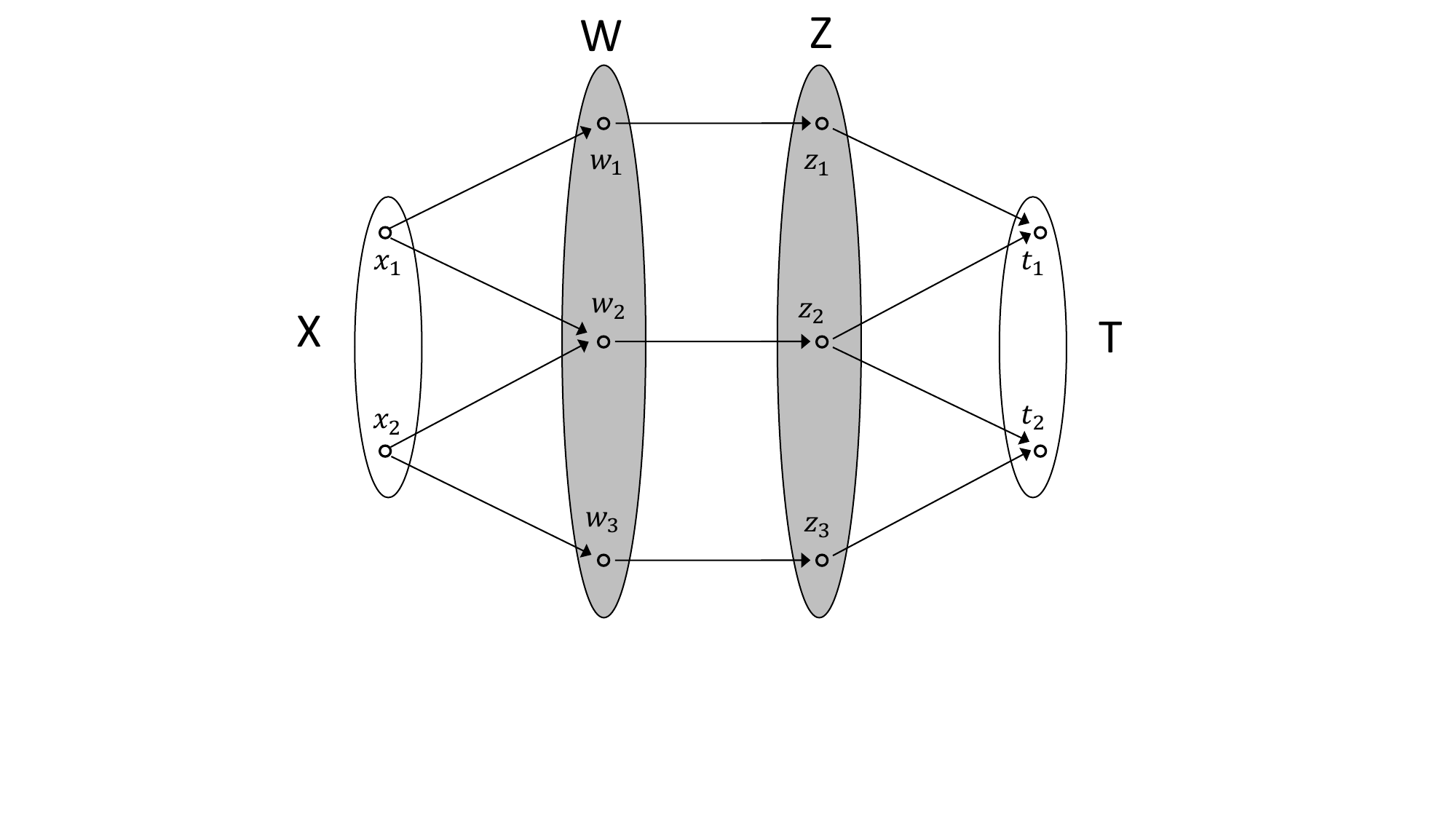}
\vspace{-20mm}
\caption{$W$ is a minimal $X-T$ separator, but it is not an important $X-T$ separator as $Z$ satisfies $|Z|=|W|$ and
 $R^{+}_{G\setminus W}(X) = X \subset X\cup W = R^{+}_{G\setminus Z}(X)$. In fact it is easy to check that the only important $X-T$ separator of size 3 is $Z$.
 If $k\geq 2$ then the set $\{z_1, z_2\}$ is in $\IS$, since it is an important $x_{1}-T$ separator of size $2$. Finally, $x_1$ belongs to the ``exact reverse shadow"
 of each of the sets $\{w_1,w_2\}, \{w_1,z_2\}, \{w_2, z_1\}$ and $\{z_1,z_2\}$, since they are all minimal $x_{1}-T$ separators. However $x_1$ does not belong to the exact reverse shadow of the set $W$ as it is not a minimal $x_{1}-T$ separator.\label{fig:important-separators}}
\end{figure}

We refer the reader to Figure~\ref{fig:important-separators} for
examples of Definitions~\ref{defn-imp-sep},
~\ref{defn-forward-reverse-impsep} and \ref{defn-exact-shadow}. Note that from the two definitions appearing in Defintion~\ref{defn-exact-shadow}, we will be using only the exact reverse shadow in the paper; the definition of exact forward shadow is given only for completeness. The
exact reverse shadow of $W$ is a subset of the reverse shadow of $W$:
it contains a vertex $v$ only if every vertex $w\in S$ is ``useful''
in separating $v$ from $T$, i.e., vertex $w$ can be reached from $v$ and $T$ can be
reached from $w$. Similarly for the forward shadow. This slight difference between the shadow and the exact shadow will be crucial in
the analysis of the algorithm (Section~\ref{sec:analysis-algorithm}).

The weaker version of the random sampling described in Section~\ref{sec:random-sampling-1}
(Theorem~\ref{th:random-sampling-bad}) randomly selects members of $\IS$ and creates a subset by taking the union of the exact reverse
shadows of these sets.  The following lemma will be used to give an upper bound on the probability that a vertex is covered
by the union.
\begin{lemma}
\label{lemma-bounding-number-of-shadows-per-vertex} Let $z$ be any vertex. Then there are at most $4^{k}$ members of $\IS$ that
contain $z$ in their exact reverse shadows.
%Furthermore we can enumerate all these shadows in time
%$O^{*}(4^{k})$.
\end{lemma}

For the proof of Lemma~\ref{lemma-bounding-number-of-shadows-per-vertex}, we need to establish first the following:

\begin{lemma}
\label{lemma-forward-impsep-exact-back-shadow} If $W\in \IS$ and $v$ is in the exact reverse shadow of $W$, then $W$ is
an important $v-T$ separator.
\end{lemma}
\begin{proof}
  Let $w$ be the witness that $W$ is in $\IS$, i.e., $W$ is an
  important $w-T$ separator in $G$. Let $v$ be any vertex in the exact
  reverse shadow of $W$, which means that $W$ is a minimal $v-T$
  separator in $G$. Suppose that $W$ is not an important $v-T$
  separator. Then there exists a $v-T$ separator $W'$ such that
  $|W'|\leq |W|$ and $R^{+}_{G\setminus W}(v)\subset R^{+}_{G\setminus
    W'}(v)$. We will arrive to a contradiction by showing that
  $R^{+}_{G\setminus W}(w)\subset R^{+}_{G\setminus W'}(w)$, i.e., $W$
  is not an important $w-T$ separator.

First, we claim that $W'$ is a $(W\setminus W')-T$ separator. Suppose that there is a path $P$ from some $x\in W\setminus W'$
to $T$ that is disjoint from $W'$. As $W$ is a minimal $v-T$ separator, there is a path $Q$ from $v$ to $x$ whose internal
vertices are disjoint from $W$. Furthermore, $R^{+}_{G\setminus W}(v)\subset R^{+}_{G\setminus W'}(v)$ implies that  the
internal vertices of $Q$ are disjoint from $W'$ as well. Therefore, concatenating $Q$ and $P$ gives a path from $v$ to $T$
that is disjoint from $W'$, contradicting the fact that $W'$ is a $v-T$ separator.

We show that $W'$ is a $w-T$ separator and its existence contradicts the assumption that $W$ is an important $w-T$ separator.
First we show that $W'$ is a $w-T$ separator. Suppose that there is a $w-T$ path $P$ disjoint from $W'$. Path $P$ has to go
through a vertex $y\in W\setminus W'$ (as $W$ is a $w-T$ separator). Thus by the previous claim, the subpath of $P$ from $y$
to $T$ has to contain a vertex of $W'$, a contradiction.

Finally, we show that $R^{+}_{G\setminus W}(w)\subseteq
R^{+}_{G\setminus W'}(w)$. As $W\neq W'$ and $|W'|\le |W|$, this will
contradict the assumption that $W$ is an important $w-T$ separator.
Suppose that there is a vertex $z \in R^{+}_{G\setminus W}(w)\setminus
R^{+}_{G\setminus W'}(w)$ and consider a $w-z$ path that is fully
contained in $R^{+}_{G\setminus W}(w)$, i.e., disjoint from $W$.  As
$z\not \in R^{+}_{G\setminus W'}(w)$, path $Q$ contains a vertex $q\in
W'\setminus W$.  Since $W'$ is a minimal $v-T$ separator, there is a
$v-T$ path that intersects $W'$ only in $q$. Let $P$ be the subpath of
this path from $q$ to $T$. If $P$ contains a vertex $r\in W$, then the
subpath of $P$ from $r$ to $T$ contains no vertex of $W'$ (as $z\neq
r$ is the only vertex of $W'$ on $P$), contradicting our earlier claim
that $W'$ is a $(W\setminus W')-T$ separator. Thus $P$ is disjoint
from $W$, and hence the concatenation of the subpath of $Q$ from $w$
to $q$ and the path $P$ is a $w-T$ path disjoint from $W$, a
contradiction.
\end{proof}

Lemma~\ref{lemma-bounding-number-of-shadows-per-vertex} easily follows from
Lemma~\ref{lemma-forward-impsep-exact-back-shadow}.  Let $J$ be a member of $\IS$ such that $z$ is in the exact reverse shadow of $J$.
By Lemma~\ref{lemma-forward-impsep-exact-back-shadow}, $J$ is an important $z-T$ separator. By Lemma~\ref{number-of-imp-sep},
there are at most $4^{k}$ important $z-T$ separators of size at most $k$ and hence $z$ belongs to at most $4^{k}$ exact reverse
shadows.

%\medskip

\begin{remark}\label{remark:imp-of-exact-shadows}
  \emph{It is crucial to distinguish between ``reverse shadow'' and
    ``exact reverse shadow'':
    Lemma~\ref{lemma-forward-impsep-exact-back-shadow} (and hence
    Lemma~\ref{lemma-bounding-number-of-shadows-per-vertex}) does not
    remain true if we remove the word ``exact.'' Consider the
    following example (see Figure~\ref{fig:imp-of-exact-new}). Let
    $a_1$, $\dots$, $a_r$ be vertices such that there is an edge going
    from every $a_i$ to every vertex of $T=\{t_1, t_2, \ldots,
    t_k\}$. For every $1\le i \le r$, let $b_i$ be a vertex with an
    edge going from $b_i$ to $a_i$. For every $1\le i < j \le r$, let
    $c_{i,j}$ be a vertex with two edges going from $c_{i,j}$ to $a_i$
    and $a_j$. Then every set $\{a_i,a_j\}$ is in $\IS$, since it is
    an important $c_{i,j}-T$ separator; and every set $\{a_i\}$ is in
    $\IS$ as well, as it is an important $b_i-T$ separator. Every
    $b_i$ is in the reverse shadow of $\{a_j, a_i\}$ for $1\leq i\neq
    j\leq r$. However, $b_i$ is in the {\em exact} reverse shadow of
    exactly one member of $\IS$, the set $\{a_i\}$.}
\end{remark}

\begin{figure}[t]
\centering
\includegraphics[width=7in]{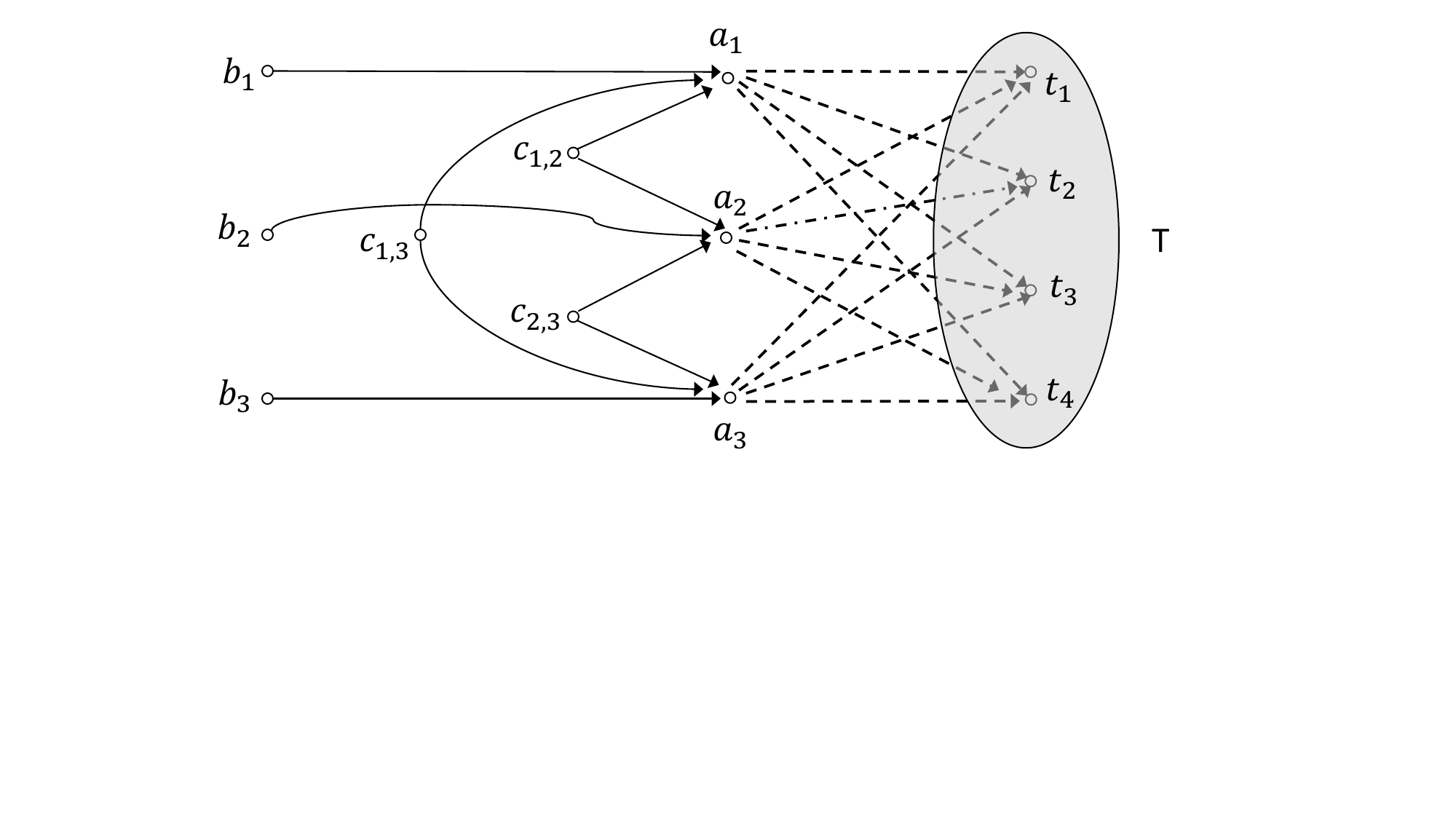}
\vspace{-50mm}
\caption{An illustration of Remark~\ref{remark:imp-of-exact-shadows} in the special case when $k=4$ and $r=3$. \label{fig:imp-of-exact-new}}
\end{figure}

\subsubsection{Random sampling} \label{sec:random-sampling-1}

In this subsection, we describe the technique of random sampling of important separators, which is crucial to the proof of
Theorem~\ref{thm:main-covering-shadow}. This technique was introduced in~\cite{marx-razgon-stoc-11} and was adapted to directed
graphs in~\cite{directed-multiway-cut}. We follow it closely and try to present it in a self-contained way that might be
useful for future applications. In Section~\ref{sec:torso}, in order to reduce the problem (via the ``torso" operation) to a shadowless instance, we
need a set $Z$ that has the following property:
\begin{center}
\noindent\framebox{\begin{minipage}{6.00in}
\textbf{Property (*)}\\
There is an $\F$-transversal $T^*$ of size at most $k$ such that $Z$ covers the shadow of $T^*$, but $Z$ is disjoint from $T^*$.
\end{minipage}}
\end{center}
%\begin{quote}
%There is a $\F$-transversal $T^*$ such that $Z$ covers the shadow of $T^*$, but $Z$ is disjoint from $T^*$. \hfill (*)
%\end{quote}

Of course, when we are trying to construct this set $Z$, we do not
know anything about the $\F$-transversals of the instance. In
particular we have no way of checking if a given set $Z$ satisfies
this property. Nevertheless, we use a randomized procedure that
creates a set $Z$ and we give a lower bound on the probability that
$Z$ satisfies the requirements. For the construction of this set $Z$,
one can use a very specific probability distribution that was
introduced in \cite{marx-razgon-stoc-11}. This probability
distribution is based on randomly selecting ``important separators''
and taking the union of their shadows. In this paper, we modify the
selection of important separators in a way that improves the success
probability.  The precise description of the randomized procedure and
the properties of the distribution it creates is described in
Theorems~\ref{th:random-sampling-bad} and \ref{th:random-sampling}.
Using the theory of splitters we can derandomize the randomized selection into a deterministic
algorithm that returns a bounded number of sets such that at least one
of them satisfies the required property
(Section~\ref{subsection-derandomization}).

First we focus on the reverse shadow and try to ensure that (with good
probability) $Z$ covers the reverse shadow of $T^*$. Then in
Section~\ref{sec:analysis-algorithm}, we argue that, after reversing
the orientation of the edges of the graph, a second application of the
random selection can be used to cover the forward shadow. Thus in this
section, we consider only the reverse shadow of $T^*$.

Roughly speaking, we want to select a random set $Z$ such that for every every $(W,Y)$ where $Y$ is in the reverse shadow of $W$,
the probability that $Z$ is disjoint from $W$ but contains $Y$ can be bounded from below. We can guarantee such a lower bound
only if $(W,Y)$ satisfies two conditions. First, it is not enough that $Y$ is in the shadow of $W$ (or in other words, $W$ is
an $Y-T$ separator), but $W$ should contain important separators separating the vertices of $Y$ from $T$ (see
Theorems~\ref{th:random-sampling-bad} and \ref{th:random-sampling} for the exact statement). Second, $W$ and $Y$ have to be disjoint, otherwise
there is clearly no set covering $Y$ and disjoint from $W$.
In other words, a vertex of $W$ cannot
be in the reverse shadow of other vertices of $W$, which is expressed by the following technical definition:

\begin{definition}{\bf (thin)}
\label{defn-thin-set} Let $G$ be a directed graph. We say that a set $W\subseteq V(G)$ is \emph{thin} in $G$ if there is no
$v\in W$ such that $v$ belongs to the reverse shadow of $W\setminus v$ with respect to $T$.
\end{definition}

We first give an easy version of the random sampling, which only gives a double exponentially small lower bound on the
probability of constructing a set $Z$ with the required properties.

\begin{theorem}{\bf (random sampling)}
\label{th:random-sampling-bad} There is an algorithm $\randset(G,T,k)$ that produces a random set $Z\subseteq V(G)\setminus T$
in time $O^*(4^k)$ such that the following holds. Let $W$ be a \emph{thin} set with $|W|\le k$, and let $Y$ be a set such that
for every $v\in Y$ there is an important $v-T$ separator $W'\subseteq W$. For every such pair $(W,Y)$, the probability that
the following two events both occur is $2^{-2^{O(k)}}$:
\begin{enumerate}
\item $W\cap Z=\emptyset$, and
\item $Y\subseteq Z$.
\end{enumerate}
\end{theorem}
\begin{proof}

The algorithm $\randset(G,T,k)$ first enumerates the collection $\IS$; let $\X$ be the set of all exact reverse
shadows of these sets. Note that two different sets in $\IS$ have different exact reverse shadows: if $X$ is the exact reverse shadow of $J\in \IS$, then $J$ is exactly the set of vertices not in $X$ and having an inneighbor in $X$.
 By Lemma~\ref{number-of-imp-sep}, the size of $\X$ is $O^*(4^k)$ and can be constructed in time
$O^*(4^k)$. Let $\mathcal{X'}$ be the subset of $\X$ where each element from $\mathcal{X}$ occurs with probability
$\frac{1}{2}$ independently at random. Let $Z$ be the union of the exact reverse shadows in $\mathcal{X'}$. We claim that the
set $Z$ satisfies the requirement of the theorem.

Let us fix a pair $(W,Y)$ as in the statement of the theorem.  Let $X_1,X_2,\ldots,X_d\in \X$ be the exact reverse shadows of
every member of $\IS$ that is a subset of $W$. As $|W|\leq k$, we have $d\le 2^k$.  By the assumption that $W$ is \emph{thin}, we have
$X_j\cap W=\emptyset$ for every $j\in[d]$.  Now consider the following events:
\begin{enumerate}
\item[(E1)]$W\cap Z=\emptyset$
\item[(E2)] $X_{j}\subseteq Z$ for every $j\in [d]$
\end{enumerate}
First we show that (E2) implies that $Y\subseteq Z$: $v\in Y$ implies there is an important separator $W'\subseteq W$, i.e., there is some $\ell\in [d]$ such that $X_{\ell}$ is the exact reverse shadow of $W$. Also note that $v\in X_{\ell}$ since $W'$ is a minimal (in fact important) $v-T$ separator. Since $X_j\subseteq Z$ for every $j\in [d]$, we have that $v\in Z$. This shows that $Y\subseteq Z$.

Our goal is to show that both events (E1) and (E2) occur with probability $2^{-2^{O(k)}}$. Let $A=\{X_1,X_2,\ldots,X_d\}$ and $B=\{X\in \X\ |\ X\cap W \neq \emptyset \}$. By Lemma
\ref{lemma-bounding-number-of-shadows-per-vertex}, each vertex of $W$ is contained in the exact reverse shadows of
at most $4^k$ members of $\IS$. Thus $|B|\leq |W|\cdot 4^{k}\leq k\cdot 4^{k}$. If no exact reverse shadow from $B$ is selected, then event (E1)
holds. If every exact reverse shadow from $A$ is selected, then event (E2) holds. Thus the probability that both (E1) and (E2)
occur is bounded from below by the probability of the event that every element from $A$ is selected and no element from $B$ is
selected. Note that $A$ and $B$ are disjoint: $A$ contains only sets disjoint from $W$, while $B$ contains only sets
intersecting $W$. Therefore, the two events are independent and the probability that both events occur is at least
\[
\Big(\frac{1}{2}\Big)^{2^{k}}\Big(1-\frac{1}{2}\Big)^{k\cdot 4^{k}} = 2^{-2^{O(k)}}
\]
\end{proof}

We now give an improved version of the random sampling that gives a stronger lower bound
on the success probability than the one guaranteed by Theorem~\ref{th:random-sampling-bad}.
%
%
%We need the following simple property of thin sets: if $M$ is a thin set, then it cannot contain any member of $\IS$ whose exact reverse
%shadow intersects $M$.\daniel{This statement is very out of place here, it would be better in the following proof.}
%
%\begin{lemma}
%Let $M$ be a thin set. If $X$ is the exact reverse shadow of some $J\in \IS$ such that $X\cap M\neq \emptyset$, then $J\nsubseteq
%M$. \label{lem:imp-of-thin}
%\end{lemma}
%\begin{proof}
%Let $X\cap M = M'$. Now $X$ is the exact reverse shadow of $J$ implies that $J$ is a minimal $X-T$ separator. If $J\subseteq
%M$, then this means that $M\setminus X$ is also an $X-T$ separator, i.e., $M'$ lies in the reverse shadow of $M\setminus M'$
%contradicting the fact that $M$ is a thin set.
%\end{proof}
%
Recall that in
Theorem~\ref{th:random-sampling-bad}, we randomly selected members of $\IS$ and took $Z$ as the union of the exact reverse shadows of
the selected sets. However, we only had single-exponential upper bounds on both types of exact reverse shadows: number of
shadows intersecting $W$ was at most $k\cdot 4^k$ and the number of exact reverse shadows of every subset of
$W$ is at most $2^k$. In Theorem~\ref{th:random-sampling}, we take a different view: we randomly select a subset of vertices
$P$ and take $Z$ as the union of exact reverse shadows of every subset of $\mathcal{P}$. This will give us a
stronger (single exponentially small) lower bound on the probability that the constructed set $Z$ satisfies the required
properties.

\begin{theorem}{\bf (improved random sampling)}
\label{th:random-sampling} There is an algorithm $\randset(G,T,k)$ that produces a random set $Z\subseteq V(G)\setminus T$ in
time $O^*(4^k)$ such that the following holds. Let $W$ be a \emph{thin} set with $|W|\le k$, and let $Y$ be a set such that
for every $v\in Y$ there is an important $v-T$ separator $W'\subseteq W$. For every such pair $(W,Y)$, the probability that
the following two events both occur is $2^{-O(k^2)}$:
\begin{enumerate}
\item $W\cap Z=\emptyset$, and
\item $Y\subseteq Z$.
\end{enumerate}
\end{theorem}
\begin{proof}
The algorithm $\randset(G,T,k)$ picks a subset $P$ of $V(G)$
where each element occurs with probability $4^{-k}$ uniformly at
random.  For every $S\in \IS$ with $S\subseteq P$, let us add the
exact reverse shadow of $S$ to $\mathcal{X'}$. Let $Z$ be the union of
the exact reverse shadows in $\mathcal{X'}$. We claim that the set $Z$
satisfies the requirement of the theorem.

Fix a pair $(W,Y)$ as in the statement of the theorem.
For each $w\in W$, we define
\begin{align*}
\mathcal{L}_w&= \{S\  |\ \text{$S$ is an important $w-T$ separator of size $\leq k$}\},\\
I_w& = \bigcup_{S\in \mathcal{L}_w} S,\\
I&=\bigcup_{w\in W}I_w.
\end{align*}
Note that a vertex $w\in W$ may have an outneighbor in $T$, in which case $\mathcal{L}_w$ and $I_w$ are empty.
Since $|W|\leq k$ and for each $w\in W$ there are at most $4^k$
important $w-T$ separators of size at most $k$, we have $|I_w|\leq
k\cdot 4^{k}$. Since $|W|\leq k$, we have $|I|\leq k^{2}\cdot
4^k$.

 Let $\X$ be the
set of exact reverse shadows of every set $S\in \IS$.
%By Lemma~\ref{number-of-imp-sep}, the size of $\X$ is $O^*(4^k)$ and can be constructed in time $O^*(4^k)$.
Let $X_1,X_2,\ldots,X_d\in \X$ be the exact reverse shadows of every
$S\in \IS$ with $S\subseteq W$. Let $A=\{X_1,X_2,\ldots,X_d\}$ and
$B=\{X\in \X\ |\ X\cap W \neq \emptyset \}$.
%As $|W|\leq k$, we have $d\le 2^k$.  By assumption that $W$ is \emph{thin}, we have
%$X_j\cap W=\emptyset$ for every $j\in[d]$.
Now consider the following events:
\begin{enumerate}
\item[(E1)]$W\cap Z=\emptyset$
\item[(E2)] $X_{j}\subseteq Z$ for every $j\in [d]$
\end{enumerate}
First we show that (E2) implies that $Y\subseteq Z$: $v\in Y$ implies there is an important separator $W'\subseteq W$, i.e., there is some $\ell\in [d]$ such that $X_{\ell}$ is the exact reverse shadow of $W$. Also note that $v\in X_{\ell}$ since $W'$ is a minimal (in fact important) $v-T$ separator. Since $X_j\subseteq Z$ for every $j\in [d]$, we have that $v\in Z$. This shows that $Y\subseteq Z$.

Our goal is to show that both events (E1) and (E2) occur with probability
$2^{-O(k^2)}$. If every vertex from $W$ is selected in $P$, then every reverse shadow from $A$ is selected into $\mathcal{X}'$
and event (E2) holds. We claim that if no vertex from $I\setminus W$ is selected in $P$, then no exact reverse shadow from $B$
is selected into $\mathcal{X}'$ and hence event $(E1)$ will also hold. Suppose to the contrary that an exact reverse shadow
$X\in B$ was selected into $\mathcal{X}'$; by the definition of $B$, there is a vertex $w\in X\cap W$. Let $J\in \IS$ be the set whose exact reverse shadow is $X$,
which implies by Lemma~\ref{lemma-forward-impsep-exact-back-shadow} that $J\in\mathcal{L}_w$ and $J\subseteq I_w\subseteq I$.
If $J\setminus W\neq \emptyset$, then the assumption that no
vertex of $I\setminus W$ was selected into $P$ condtradicts the fact that $X$ was selected into $\mathcal{X}'$. Suppose therefore that  $J\subseteq W$ holds. Since $X$ is the exact reverse shadow of $J$, we know that $J$ is a minimal $X-T$ separator. But $J\subseteq W$ implies that $W\setminus X$ is also an $X-T$ separator, i.e., $W\cap X$ lies in the reverse shadow of $W\setminus X$. This contradicts the fact that $W$ is a thin set (see Definition~\ref{defn-thin-set}).
%But $X\in B$ implies $X\cap W\neq \emptyset$ and hence by Lemma~\ref{lem:imp-of-thin}
%we have that $J\in \IS$ whose exact reverse shadow is $X$ cannot be contained in $W$. This contradicts the fact that no
%vertex of $I\setminus W$ was selected in $P$.

Thus the probability that both the events (E1) and (E2) occur is bounded from below by the probability of the event that every
vertex from $W$ is selected in $P$ and no vertex from $I\setminus W$ is selected in $P$. Note that the sets $W$ and
$I\setminus W$ are clearly disjoint. Therefore, the two events are independent and the probability that both events occur is
at least
$$
(4^{-k})^{k}(1-4^{-k})^{k^{2}\cdot 4^{k}} \geq 4^{-k^{2}}\cdot e^{-2k^{2}}  = 2^{-O(k^2)}
$$
where we used the inequalities that $1+x \geq e^{\frac{x}{1+x}}$ for every $x> -1$ and $1-4^{-k}\geq \frac{1}{2}$ for every
$k\geq 1$.
\end{proof}

\subsection{Derandomization} \label{subsection-derandomization}

We now derandomize the process of choosing exact reverse shadows in Theorem~\ref{th:random-sampling} using the technique of
\emph{splitters}. An $(n,r,r^2)$-splitter is a family of functions from $[n]\rightarrow [r^2]$ such that for every $M\subseteq
[n]$ with $|M|=r$, at least one of the functions in the family is injective on $M$. Naor et al.~\cite{aravind-focs-1995} give an explicit construction of an $(n,r,r^2)$-splitter of size
$O(r^{6}\log r\log n)$ in time $\text{poly}(n,r)$.

\begin{theorem}{\bf (deterministic sampling)}
\label{th:derandom-sampling} There is an algorithm $\derandset(G,T,k)$ that produces $t=2^{O(k^2)}\log |V(G)|$ subsets of $Z_1$, $\dots$, $Z_t$ of $V(G)\setminus T$  in
time $O^*(2^{O(k^2)})$ such that the following holds. Let $W$ be a \emph{thin} set with $|W|\le k$, and let $Y$ be a set such that
for every $v\in Y$ there is an important $v-T$ separator $W'\subseteq W$. For every such pair $(W,Y)$, there is at least one $1\le i \le t$ with
\begin{enumerate}
\item $W\cap Z=\emptyset$, and
\item $Y\subseteq Z$.
\end{enumerate}
\end{theorem}
\begin{proof}
In the proof of Theorem~\ref{th:random-sampling}, a random subset $P$ of a universe $V(G)$ of size $n$ is selected. We argued
that if every vertex from $W$ is selected in $P$ and no element from $I\setminus W$ is selected, then both the events (E1) and
(E2) occur. Instead of selecting a random subset $P$, we will construct several subsets such that at least one of them will
contain every vertex in $W$ and no vertex from $I\setminus W$. Let $n=|V(G)|$, $a=|W|\leq k$, and $b=|I\setminus W|\leq k^{2}\cdot 4^{k}$. Each subset is defined by a
pair $(h,H)$, where $h$ is a function in an $(n,a+b,(a+b)^2)$-splitter family and $H$ is a subset of $[(a+b)^2]$ of size $a$
(there are $\binom{(a+b)^2}{a} = \binom{(k+k^{2}\cdot 4^{k})^2}{k} = 2^{O(k^2)}$ such sets $H$). For a particular choice of
$h$ and $H$, we select those vertices $v\in V(G)$ into $P$ for which $h(v)\in H$. The size of the splitter family is
$O\Big((a+b)^{6}\log(a+b)\log(n)\Big) = 2^{O(k)}\log n$ and the number of possibilities for $H$ is $2^{O(k^2)}$. Therefore, we
construct $2^{O(k^2)}\log n$ subsets of $V(G)$. The total time taken for constructing these subsets is $\text{poly}(n,a+b) = \text{poly}(n,4^k)$.

By the definition of the splitter, there is a function $h$ that is injective on $W$, and there is a subset $H$ such that
$h(v)\in H$ for every set $v\in W$ and $h(y)\not\in H$ for every $y\in I\setminus W$. For such an $h$ and $H$, the selection
will ensure that (E1) and (E2) hold. Thus at least one of the constructed subsets has the required properties, which is what
we had to show.
\end{proof}

\subsection{Proof of Theorem~\ref{thm:main-covering-shadow}: The \textsc{Covering} Algorithm}
\label{sec:analysis-algorithm}

To prove Theorem~\ref{thm:main-covering-shadow}, we show that Algorithm~\ref{alg:covering} gives a set $Z$ satisfying the properties of Theorem~\ref{thm:main-covering-shadow}.  Due to the delicate way separators and shadows behave in directed graphs, we construct the set $Z$ in two phases, calling the function $\randset$ of Section~\ref{subsection-imp-separators} twice and taking $Z$ to be the union of the two outputs. For consistency of notation, we denote the input graph by $G_1$. Let $Z_1$ be the output of the first call of the function $\randset$, i.e., $Z_1=\randset(G_1,T,k)$. We build a new graph $G_2$ from $G_1$ by reversing the orientation of every edge and adding every vertex of $Z_1$ to $V^{\infty}$. Since the structure of the graph $G_2$ depends on the set $Z_1$, the distribution of the second random sampling depends on the result $Z_1$ of the first random sampling. This means that we cannot make the two calls in parallel. Our aim is to show that there is a transversal $T^*$ such that we can give a lower bound on the probability that $Z_1$ covers $r_{G_1,T}(T^*)$ and $Z_2$ covers $f_{G_1,T}(T^*)$.

% This is formalized in the following claim:

% \begin{claim}
% \label{claim:alg-covering-is-good} The set $Z$ in the output of Algorithm~\ref{alg:covering} satisfies the conditions of
% Theorem~\ref{thm:main-covering-shadow}.
% \end{claim}

\begin{algorithm}[t]
\caption{\textsc{Covering} (randomized version)} \label{alg:covering}
\textbf{Input:} A directed graph $G_1$, integer $k$.\\
\textbf{Output:} A set $Z$.
\begin{algorithmic}[1]

\STATE Let $Z_1=\randset(G_1,T,k)$.

\STATE Let $G_2$ be obtained from $G_1$ by reversing the orientation of every edge and adding every vertex of $Z_1$ to $V^{\infty}$.% (i.e., $V^{\infty}(G_2)=V^{\infty}(G_1)\cup Z_1$).

\STATE Let $Z_2=\randset(G_2,T,k)$.

\STATE Let $Z=Z_1\cup Z_2$.
%
%\STATE Let $G_3=\torso(G_1,V(G)\setminus Z)$.

\end{algorithmic}
\end{algorithm}

To prove the existence of the required transversal $T^*$, we need the following definition:
\begin{definition}
{\bf (shadow-maximal transversal)} \label{defn-terminal-minimal-solution} An $\F$-transversal $W$ is {\em minimum} if there is no $\F$-transversal of size less than $|W|$. A minimum $\F$-transversal $W$ is called \emph{shadow-maximal} if $r_{G_1,T}(W)\cup f_{G_1,T}(W) \cup W$ is
inclusion-wise maximal among all minimum $\F$-transversals.
\end{definition}

For the rest of the proof, let us fix $T^*$ to be a shadow-maximal $\F$-transversal such that $|r_{G_1,T}(T^*)|$ is maximum
possible among all shadow-maximal $\F$-transversals. We bound the probability that $Z\cap T^*=\emptyset$ and $r_{G_1,T}(T^*)\cup
f_{G_1,T}(T^*)\subseteq Z$. More precisely, we bound the probability that all of the following four events occur:
\begin{enumerate}
\item $Z_1\cap T^*=\emptyset$,
\item $r_{G_1,T}(T^*)\subseteq Z_1$,
\item $Z_2\cap T^*=\emptyset$, and
\item $f_{G_1,T}(T^*)\subseteq Z_2$.
\end{enumerate}
That is, the first random selection takes care of the reverse shadow, the second takes care of the forward shadow, and none of
$Z_1$ or $Z_2$ hits $T^*$. Note that it is somewhat counterintuitive that we choose a $T^*$ for which the shadow is large:
intuitively, it seems that the larger the shadow is, the less likely that it is fully covered by $Z$. However, we need this
maximality property in order to bound the probability that $Z\cap T^*=\emptyset$.

We want to invoke Theorem~\ref{th:random-sampling} to bound the probability that $Z_1$ covers $Y=r_{G_1,T}(T^*)$ and $Z_1\cap
T^*=\emptyset$. First, we need to ensure that $T^*$ is a \emph{thin} set, but this follows easily from the fact that $T^*$ is
a minimum $\F$-transversal:
\begin{lemma}
\label{lemma:minimal-shadow-cover-sol} If $W$ is a minimum $\F$-transversal for some $T$-connected $\F$, then no $v\in W$ is in the reverse shadow of some
$W'\subseteq W\setminus v$.
\end{lemma}
\begin{proof}
Suppose to the contrary that there is a vertex $v\in W$ such that $v\in r(W')$ for some $W'\subseteq W\setminus v$. Then
we claim that $W\setminus v$ is also an $\F$-transversal, contradicting the minimality of $W$. Let
$\F=\{F_1,F_2,\ldots,F_q\}$ and suppose that there is a $i\in [q]$ such that $F_{i}\cap W=\{v\}$. As $\F$ is $T$-connected, there is
a $v\rightarrow T$ walk $P$ in $F_i$. But $P\cap W=\{v\}$ implies that there is a $v\rightarrow T$ walk in $G\setminus
(W\setminus v)$, i.e., $v$ cannot belong to the reverse shadow of any $W'\subseteq W\setminus v$.
\end{proof}

More importantly, if we want to use Theorem~\ref{th:random-sampling} with $Y=r_{G_1,T}(T^*)$, then we have to make sure that for
every vertex $v$ of $r_{G_1,T}(T^*)$, there is an important $v-T$ separator that is a subset of $T^*$. The ``pushing argument''
of Lemma~\ref{thm-pushing} shows that if this is not true for some $v$, then we can modify the $\F$-transversal in a way that
increases the size of the reverse shadow. The extremal choice of $T^*$ ensures that no such modification is possible, thus
$T^*$ contains an important $v-T$ separator for every $v$.

\begin{lemma}
{\bf (pushing)} \label{thm-pushing} Let $W$ be an $\F$-transversal for some $T$-connected $\F$. For every $v\in r(W)$, either there is a $W_1\subseteq W$
that is an important $v-T$ separator, or there is an $\F$-transversal $W'$ such that
\begin{enumerate}
\item $|W'|\le |W|$,
\item $r(W)\subset r(W')$,
%\item $r(S)\cup S \subset r(S')\cup S'$, and
\item $(r(W)\cup f(W) \cup W) \subseteq (r(W')\cup f(W') \cup W')$.
\end{enumerate}
\end{lemma}
\begin{proof}
Let $W_0$ be the subset of $W$ reachable from $v$ without going through any other vertices of $W$. Then $W_0$ is clearly a
$v-T$ separator. Let $W_1$ be the minimal $v-T$ separator contained in $W_0$ (we may note that if $W$ is a minimal $\F$-transversal, then we always have $W_1=W_0$).  If $W_1$ is an important $v-T$ separator, then
we are done as $W$ itself contains $W_1$. Otherwise, there exists an important $v-T$ separator $W'_{1}$, i.e., $|W'_{1}|\leq
|W_1|$ and $R^{+}_{G\setminus W_{1}}(v)\subset R^{+}_{G\setminus W'_{1}}(v)$. Now we show that $W' = (W\setminus W_1)\cup
W'_{1}$ is also an $\F$-transversal. Note that $W'_{1}\subseteq W'$ and $|W'|\leq |W|$.

First we claim that $r(W)\cup (W\setminus W')\subseteq r(W')$. Suppose that there is a walk $P$ from $\beta$ to $T$ in
$G\setminus W'$ for some $\beta\in r(W)\cup(W\setminus W')$. If $\beta\in r(W)$, then walk $P$ has to go through a vertex
$\beta'\in W$. As $\beta'$ is not in $W'$, it has to be in $W\setminus W'$. Therefore, by replacing $\beta$ with $\beta'$, we
can assume in the following that $\beta\in W\setminus W'\subseteq W_1\setminus W'_1$.  By the minimality of $W_1$, every vertex of
$W_1\subseteq W_0$ has an incoming edge from some vertex in $R^{+}_{G\setminus W}(v)$. This means that there is a vertex
$\alpha \in R^{+}_{G\setminus W}(v)$ such that $(\alpha,\beta)\in E(G)$. Since $R^{+}_{G\setminus W}(v)\subset
R^{+}_{G\setminus W'}(v)$, we have $\alpha \in R^{+}_{G\setminus W'}(v)$, implying that there is a $v\rightarrow \alpha$ walk
in $G\setminus W'$. The edge $\alpha \rightarrow \beta$ also survives in $G\setminus W'$ as $\alpha \in R^{+}_{G\setminus
W'}(v)$ and $\beta \in W_{1} \setminus W'_{1}$. By assumption, we have a walk in $G\setminus W'$ from $\beta$ to some $t\in
T$.  Concatenating the three walks we obtain a $v \rightarrow t$ walk in $G\setminus W'$, which contradicts the fact that $W'$
contains an (important) $v-T$ separator $W'_{1}$. This proves the claim. Since $W\neq W'$ and $|W|=|W'|$, the set $W_1\setminus W'_1$ is non-empty.
Thus $r(W)\subset r(W')$ follows from the claim $r(W)\cup (W\setminus W')\subseteq r(W')$.

Suppose now that $W'$ is not an $\F$-transversal. Then there is some $i\in [q]$ such that $F_{i}\cap W'=\emptyset$. As $W$ is an
$\F$-transversal, there is some $w\in W\setminus W'$ with $w\in F_i$. As $\F$ is $T$-connected, there is a $w\rightarrow T$ walk in $F_i$, which gives
a $w\rightarrow T$ walk in $G\setminus W'$ as $W'\cap F_{i}=\emptyset$. However, we have $W \setminus W' \subseteq r(W')$ (by the
claim in the previous paragraph), a contradiction. Thus $W'$ is also an $\F$-transversal.

Finally, we show that $r(W)\cup f(W)\cup W\subseteq r(W')\cup f(W')\cup W'$. We know that $r(W)\cup (W\setminus W')\subseteq
r(W')$. Thus it is sufficient to consider a vertex $v\in f(W)\setminus r(W)$. Suppose that $v\not \in f(W')$ and $v\not \in
r(W')$: there are walks $P_1$ and $P_2$ in $G\setminus W'$, going from $T$ to $v$ and from $v$ to $T$, respectively. As $v\in
f(W)$, walk $P_1$ intersects $W$, i.e., it goes through a vertex of $\beta \in W\setminus W'\subseteq r(W')$. However,
concatenating the subwalk of $P_1$ from $\beta$ to $v$ and the walk $P_2$ gives a walk from $\beta\in r(W')$ to $T$ in
$G\setminus W'$, a contradiction.
\end{proof}

Note that if $W$ is a shadow-maximal $\F$-transversal, then the $\F$-transversal $W'$ in Lemma~\ref{thm-pushing} is also a minimum $\F$-transversal and
shadow-maximal. Therefore, by the extremal choice of $T^*$, applying Lemma~\ref{thm-pushing} on $T^*$ cannot produce a
shadow-maximal $\F$-transversal $T'$ with $r_{G_1,T}(T^*)\subset r_{G_1,T}(T')$, and hence $T^*$ contains an important $v-T$
separator for every $v\in r_{G_1,T}(T^*)$. Thus by Theorem~\ref{th:random-sampling} for $Y=r_{G_1,T}(T^*)$, we get:
\begin{lemma}
\label{lemma:reverse} With probability at least $2^{-O(k^2)}$, both $r_{G_1,T}(T^*)\subseteq Z_1$ and $Z_1\cap T^*=\emptyset$
occur.
\end{lemma}

In the following, we assume that the events in Lemma~\ref{lemma:reverse} occur.  Our next goal is to bound the probability
that $Z_2$ covers $f_{G_1,T}(T^*)$. Let us define a collection $\F'$ of subgraphs of $G_2$ as follows: for every subgraph $F\in \F$ of $G_1$, let us add to $\F'$ the corresponding subgraph $F'$ of $G_2$, i.e., $F'$ is the same as $F$ with every edge reversed.
Note that $\F'$ is $T$-connected in $G_2$: the definition of $T$-connected is symmetric with respect to the orientation of the edges.
Moreover, $T^*$ is an $\F'$-transversal in $G_2$: the vertices in $T^*$
remained finite (as $Z_1\cap T^*=\emptyset$ by Lemma~\ref{lemma:reverse}), and reversing the orientation of the edges does not
change the fact that $T^*$ is a transversal. Set $T^*$ is also shadow-maximal as an $\F'$-transversal in $G_2$:
Definition~\ref{defn-terminal-minimal-solution} is insensitive to reversing the orientation of the edges and adding some vertices to $V^{\infty}$ can only decrease the set of potential transversals.  Furthermore, the forward shadow of $T^*$ in $G_2$ is
same as the reverse shadow of $T^*$ in $G_1$, that is, $f_{G_2,T}(T^*)=r_{G_1,T}(T^*)$. Therefore, assuming that the events in
Lemma~\ref{lemma:reverse} occur, every vertex of $f_{G_2,T}(T^*)$ is in $V^{\infty}$ in $G_2$. We show that now it holds that
$T^*$ contains an important $v-T$ separator in $G_2$ for every $v\in r_{G_2,T}(T^*)=f_{G_1,T}(T^*)$:

\begin{lemma}
\label{lemma-no-pushing-in-terminal-minimal-solns} If $W$ is a shadow-maximal $\F$-transversal for some $T$-connected $\F$ and every vertex of $f(W)$ belongs to $V^{\infty}$, then $W$ contains an important $v-T$ separator for every $v\in r(W)$.
\end{lemma}
\begin{proof}
Suppose to the contrary that there exists $v\in r(W)$ such that $W$ does not contain an important $v-T$ separator. Then by
Lemma~\ref{thm-pushing}, there is a another shadow-maximal $\F$-transversal $W'$. As $W$ is shadow-maximal, it follows that $r(W)\cup
f(W)\cup W= r(W')\cup f(W')\cup W'$. Therefore, the nonempty set $W'\setminus W$ is fully contained in $r(W)\cup f(W)\cup W$.
However, it cannot contain any vertex of $f(W)$ (as they are infinite by assumption) and cannot contain any vertex of $r(W)$
(as $r(W)\subset r(W')$), a contradiction.
\end{proof}

Recall that $T^*$ is a shadow-maximal $\F'$-transversal in $G_2$. In
particular, $T^*$ is a minimal $\F'$-transversal in $G_2$, hence
Lemma~\ref{lemma:minimal-shadow-cover-sol} implies that $T^*$ is thin
in $G_2$ also. Thus Theorem~\ref{th:random-sampling} can be used again (this time with
$Y=r_{G_2,T}(T^*)$) to bound the probability that
$r_{G_2,T}(T^*)\subseteq Z_2$ and $Z_2\cap T^*=\emptyset$. As the
reverse shadow $r_{G_2,T}(T^*)$ in $G_2$ is the same as the forward
shadow $f_{G_1,T}(T^*)$ in $G$, we can state the following:
\begin{lemma}
\label{lemma:forward} Assuming the events in Lemma~\ref{lemma:reverse} occur, with probability at least $2^{-O(k^2)}$ both
$f_{G_1,T}(T^*)\subseteq Z_2$ and $Z_2\cap T^*=\emptyset$ occur.
\end{lemma}

Therefore, with probability $(2^{-O(k^2)})^2$, the set $Z_1\cup Z_2$
covers $f_{G_1,T}(T^*)\cup r_{G_1,T}(T^*)$ and it is disjoint from
$T^*$. This completes the proof of
Theorem~\ref{thm:main-covering-shadow}.
%By Lemma~\ref{defn-reduced-instance}, this means that $T^*$ is a shadowless solution of the reduced
%instance $I/(Z_1\cup Z_2)$.

%\begin{lemma}
%With probability $2^{-2^{O(k)}}$, $S^*$ is a shadowless solution of $(G_3,T,p)$ and a solution of the undirected instance
%$(G^*_3,T,p)$.
%\end{lemma}

%\subsection{Proof of Theorem~\ref{thm:main-covering-shadow-derandomized} via Derandomization of Algorithm~\ref{alg:covering}}
%\label{sec:derandomize-covering-alg}

Finally, to prove Theorem~\ref{thm:main-covering-shadow-derandomized}, the
derandomized version of
Theorem~\ref{thm:main-covering-shadow-derandomized}, we use the
deterministic variant $\derandset(G,T,k)$ of the function $\randset(G,T,k)$ that, instead
of returning a random set $Z$, returns a deterministic set $Z_1$,
$\dots$, $Z_t$ of $t=2^{O(k^2)}\log n$ sets in $\text{poly}(n,4^k)$ time
(Theorem~\ref{th:derandom-sampling}).  Therefore, in Steps 1 and 3 of
Algorithm~\ref{alg:covering}, we can replace $\randset$ with this
deterministic variant $\derandset$, and branch on the choice of one $Z_i$ from the
returned sets. By the properties of the deterministic algorithm, if
$I$ is a yes-instance, then $Z$ has Property (*) in at least one of
the $2^{O(k^2)}\log^2 n$ branches. The branching increases the running time
only by a factor of $(O^*(2^{O(k^2)}))^2$ and therefore the total
running time is $O^{*}(2^{O(k^2)})$. This completes the proof of
Theorem~\ref{thm:main-covering-shadow-derandomized}.

%%% Local Variables:
%%% mode: latex
%%% TeX-master: "dsfvs-arxiv"
%%% End:

%% file: mway1.pdf_tex
%% Creator: Inkscape 0.48.3.1, www.inkscape.org
%% Accompanies image file 'mway1.pdf' (pdf, eps, ps)
\begingroup%
  \makeatletter%
  \providecommand\color[2][]{%
    \errmessage{(Inkscape) Color is used for the text in Inkscape, but the package 'color.sty' is not loaded}%
    \renewcommand\color[2][]{}%
  }%
  \providecommand\transparent[1]{%
    \errmessage{(Inkscape) Transparency is used (non-zero) for the text in Inkscape, but the package 'transparent.sty' is not loaded}%
    \renewcommand\transparent[1]{}%
  }%
  \providecommand\rotatebox[2]{#2}%
  \ifx\svgwidth\undefined%
    \setlength{\unitlength}{304.93551458bp}%
    \ifx\svgscale\undefined%
      \relax%
    \else%
      \setlength{\unitlength}{\unitlength * \real{\svgscale}}%
    \fi%
  \else%
    \setlength{\unitlength}{\svgwidth}%
  \fi%
  \global\let\svgwidth\undefined%
  \global\let\svgscale\undefined%
  \makeatother%
  \begin{picture}(1,0.91819394)%
    \put(0,0){\includegraphics[width=\unitlength]{mway1.pdf}}%
    \put(0.46867379,0.32182874){\color[rgb]{0,0,0}\makebox(0,0)[b]{\smash{$W$}}}%
    \put(0.20865067,0.8935832){\color[rgb]{0,0,0}\makebox(0,0)[b]{\smash{$r(W)$}}}%
    \put(0.77532786,0.8935832){\color[rgb]{0,0,0}\makebox(0,0)[b]{\smash{$f(W)$}}}%
    \put(0.48149524,0.55515099){\color[rgb]{0,0,0}\makebox(0,0)[b]{\smash{$f(W)\cap r(W)$}}}%
    \put(0.46867379,0.06997221){\color[rgb]{0,0,0}\makebox(0,0)[b]{\smash{$T$}}}%
  \end{picture}%
\endgroup%

%% file: reducing-instance.tex
\section{\textsc{Disjoint \sdfvs\ Compression}: Reduction to Shadowless Solutions}
\label{sec:torso}

We use the algorithm of Theorem~\ref{thm:main-covering-shadow-derandomized}  to construct a set $Z$ of vertices that we want
to get rid of. The second ingredient of our algorithm is an operation that removes a set of vertices without making the
problem any easier. This transformation can be conveniently described using the operation of taking the \emph{torso} of a
graph. From this point onwards in the paper, we do not follow~\cite{directed-multiway-cut}. In particular, the \emph{torso}
operation is problem-specific. For \textsc{Disjoint \sdfvs\ Compression}, we define it as follows:

\begin{definition}{\bf (torso)}
\label{defn-torso} Let $(G,S,T,k)$ be an instance of \textsc{Disjoint \sdfvs\ Compression} and $C\subseteq V(G)$. Then
$\torso(G,C,S)$ is a pair $(G',S')$ defined as follows:
\begin{itemize}
\item $G'$  has vertex set $C$ and there is (directed) edge $(a,b)$ in $G'$ if there is an $a\rightarrow b$ walk in
$G$ whose internal vertices are not in $C$,
\item $S'$ contains those edges of $S$ whose endpoints are both in $C$; furthermore, we add the edge $(a,b)$ to $S'$ if there is an $a\rightarrow b$ walk
in $G$ that contains an edge from $S$ and whose internal vertices are not in $C$.
\end{itemize}
\end{definition}

%\daniel{We need a figure for the definition of torso.}

\begin{figure}[t]
\centering
\def\svgwidth{0.9\linewidth}%
\executeiffilenewer{torso.svg}{torso.pdf}%
{inkscape -z -D --file=torso.svg %
--export-pdf=torso.pdf --export-latex}%
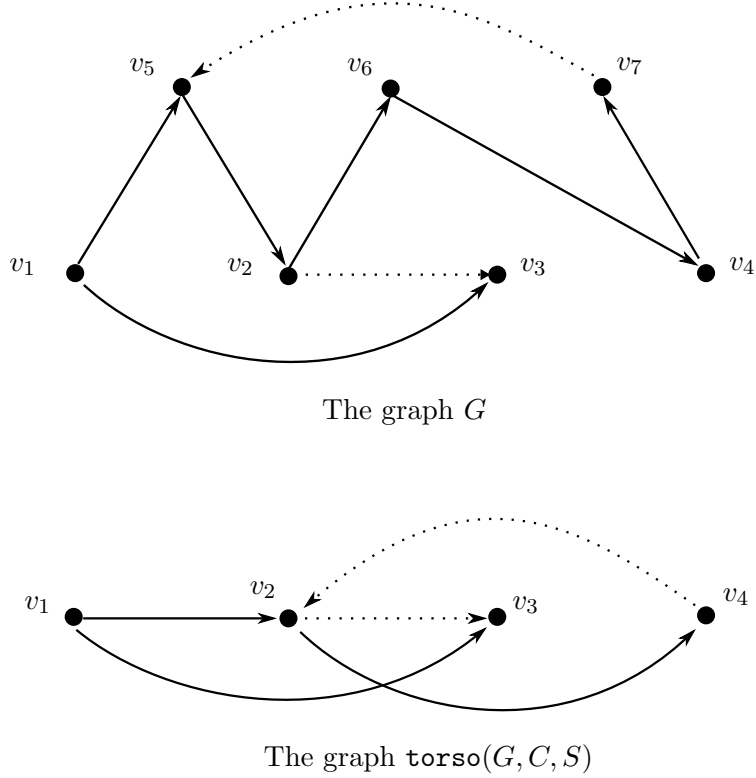%

\vspace{-90mm}
\caption{In the top graph $G$ we have $C=\{v_1, v_2, v_3, v_4\}$. The edges in $S$ are given by the dotted lines. In the bottom graph we show the graph $\texttt{torso}(G,C,S)$. All edges from $G[C]$ appear in this graph. In addition, we also add the edges $(v_1, v_2), (v_2, v_4)$ and $(v_4, v_2)$. The edge $(v_2, v_3)\in G[C]\cap S$ appears in $S'$. In addition, we also add the edge $(v_4, v_2)$ to $S'$ since the $v_4 \rightarrow v_7\rightarrow v_5\rightarrow v_2$ path in $G$ has an edge $(v_7, v_5)\in S$.\label{fig:torso}}
\end{figure}

In particular, if $a,b\in C$ and $(a,b)$ is a directed edge of $G$ and $\torso(G,C,S)=(G',S')$, then $G'$ contains $(a,b)$ as well. Thus
$G'$ is a supergraph of the subgraph of $G$ induced by $C$. Figure~\ref{fig:torso} illustrates the definition of \texttt{torso} with an example.

The following easy statement was proved in
\cite{directed-multiway-cut}: the $\torso$ operation preserves whether
a set $W\subseteq C$ is a separator.
\begin{lemma}[\cite{directed-multiway-cut}]
{\bf (torso preserves separation)} \label{torso-preserves-separation} Let $G$ be a directed graph and $C\subseteq V(G)$. Let
$(G',\emptyset)=\emph{\torso}(G,C,\emptyset)$ and $W\subseteq C$. For $a,b \in C\setminus W$, the graph $G\setminus W$ has an $a\rightarrow b$ path if
and only if $G'\setminus W$ has an $a\rightarrow b$ path.
\end{lemma}

We need a very similar statement here: the \torso\ operation preserves whether a set $W\subseteq C$ hits every $S$-closed-walk.
\begin{lemma}%$[\star]$
{\bf (torso preserves $S$-closed-walks)} \label{torso-preserves-closed-walks} Let $G$ be a directed graph with $C\subseteq
V(G)$ and $S\subseteq E(G)$. Let $(G',S')=\torso(G,C,S)$, $v\in C$, and $W\subseteq C$. Then $G\setminus W$ has an $S$-closed-walk passing through $v$ if and
only if $G'\setminus W$ has an $S'$-closed-walk passing through $v$.
\end{lemma}
\begin{proof}
Let $P$ be an $S$-closed-walk in $G\setminus W$ passing through $v$. If $P$ is fully contained in $C$, then $P$ also appears in $G'\setminus W$. Otherwise, $P$ contains vertices
from both $C$ and $V(G)\setminus C$. Let $u,w$ be two vertices of $C$ such that every vertex of $P$ between $u$ and $w$ is from
$V(G)\setminus C$. Then, by definition of torso, there is an edge $(u,w)$ in $G'$. Using such edges, we can modify $P$
to obtain another closed walk say $P'$ passing through $v$ that lies completely in $G'$ but avoids $W$. Note that
since $P$ is a $S$-closed-walk, at least one of the edges on some $u\rightarrow w$ walk that we short-circuited above must
have been from $S$ and by Definition~\ref{defn-torso} we would have put the edge $(u,w)$ edge into $S'$,
which makes $P'$ an $S'$-closed-walk in $G'$.
%Let $P$ be a path from $a$ to $b$ in $G$. Suppose $P$ is disjoint from $S$. Then $P$ contains vertices
%from $C$ and $V(G)\setminus C$. Let $u,v$ be two vertices of $C$ such that every vertex of $P$ between $u$ and $v$ is from
%$V(G)\setminus C$. Then by definition there is an edge $(u,v)$ in $\torso(G,C)$. Using such edges we can modify $P$ to obtain
%an $a\rightarrow b$ path that lies completely in $\torso(G,C)$ but avoids $S$.

Conversely, suppose that $P'$ is an $S'$-closed-walk passing through a vertex $v$ in $G'$ and it avoids $W\subseteq C$. If
$P'$ uses an edge $(u,w)\notin E(G)$, then this means that there is a $u\rightarrow w$ walk $P_{uw}$ whose internal vertices are
not in $C$. Using such walks, we modify $P'$ to get a closed walk $P$ passing through $v$ that only uses edges from $G$,
i.e., $P$ is a closed walk in $G\setminus W$. It remains to show that $P$ is an $S$-closed-walk: since $P'$ is an
$S'$-closed-walk, either some edge of $P'$ was originally in $S$ or there exist some $a,b\in P'$ such that there is a
$a\rightarrow b$ walk does not contain any vertex from $C$ and some edge on this walk was originally in $S$.
%Conversely suppose $P'$ is an $S$-cycle passing through a vertex $a$ in $\torso(G,C)$ and it avoids $W\subseteq C$. If $P'$
%uses an edge $(u,v)\notin E(G)$, then this means that there is a $u\rightarrow v$ path $P''$ whose internal vertices are not
%in $C$. Using such paths we modify $P$ to get an $a\rightarrow b$ path $P_0$ that only uses edges from $G$. Since $S\subseteq
%C$ we have that the new vertices on the path are not in $S$ and so $P_0$ avoids $S$.
\end{proof}

If we want to remove a set $Z$ of vertices, then we create a new instance by taking the torso on the {\em complement} of $Z$:
\begin{definition}
  \label{defn-reduced-instance} Let $I=(G,S,T,k)$ be an instance of
  \textsc{Disjoint \sdfvs\ Compression} and $Z\subseteq V(G)\setminus
  T$. The reduced instance $I/Z=(G',S',T,k)$ is obtained by setting
  $(G',S')=\torso(G,V(G)\setminus Z,S)$.
\end{definition}
The following lemma states that the operation of taking the torso does not make the \textsc{Disjoint \sdfvs\ Compression}
problem easier for any $Z\subseteq V(G)\setminus T$ in the sense that any solution of the reduced instance $I/Z$ is a solution
of the original instance $I$. Moreover, if we perform the torso operation for a $Z$ that is large enough to cover the shadow
of some solution $T^*$ and also small enough to be disjoint from $T^*$, then $T^*$ becomes a shadowless solution for the
reduced instance $I/Z$.

\begin{lemma}\label{reduced-instance-soln}
{\bf (creating a shadowless instance)} Let $I=(G,S,T,k)$ be an instance of \textsc{Disjoint \sdfvs\ Compression} and $Z\subseteq
V(G)\setminus T$.
\begin{enumerate}
\item If $I$ is a no-instance, then the reduced instance $I/Z$ is also a no-instance.
\item If $I$ has solution $T'$ with $f_{G,T}(T')\cup r_{G,T}(T')\subseteq Z$ and $T'\cap Z=\emptyset$, then $T'$ is a
    shadowless solution of $I/Z$.
\end{enumerate}
\end{lemma}
\begin{proof}
Let $C=V(G)\setminus Z$ and $(G',S')=\torso(G,C,S)$. To prove the first statement, suppose that
$T'\subseteq V(G')$ is a solution for $I/Z$. We show that $T'$ is also a solution for $I$. Suppose to the contrary that $G\setminus T'$ has an
$S$-closed-walk, which has to pass through some vertex $v\in T$ (since $G\setminus T$ has no
$S$-closed-walks). Note that $v\in T$ and $Z\subseteq V(G)\setminus T$ implies $v\in C$. Then by
Lemma~\ref{torso-preserves-closed-walks}, $G'\setminus T'$ also has an $S'$-closed-walk passing through $v$ contradicting the
fact that $T'$ is a solution for $I/Z$.

For the second statement, let $T'$ be a solution of $I$ with $T'\cap Z = \emptyset$ and $f_{G,T}(T')\cup r_{G,T}(T')\subseteq
Z$. We claim $T'$ is a solution of $I/Z$ as well. Suppose to the contrary that $G'\setminus T'$ has an $S'$-closed-walk passing
through some vertex $v\in C$. As $v\in C$, Lemma~\ref{torso-preserves-closed-walks} implies $G\setminus T'$ also has an
$S$-closed-walk passing through $v$, which is a contradiction as $T'$ is a solution of $I$.

Finally, we show that $T'$ is a shadowless solution, i.e,
$r_{G',T}(T')=f_{G',T}(T')=\emptyset$. We only prove
$r_{G',T}(T')=\emptyset$: the argument for $f_{G',T}(T')=\emptyset$ is
analogous.  Assume to the contrary that there exists $w\in
r_{G',T}(T')$ (note that we have $w\in V(G')$, i.e., $w\notin
Z$). This means that $T'$ is a $w-T$ separator in $G'$, i.e., there is
no $w-T$ walk in $G'\setminus T'$.  By
Lemma~\ref{torso-preserves-separation}, it follows that there is no
$w-T$ walk in $G\setminus T'$ either, i.e., $w\in r_{G,T}(T')$. But
$r_{G,T}(T')\subseteq Z$ and therefore we have $w\in Z$, which is
a contradiction.
% Thus $r_{G,T}(T')\subseteq Z$ in $G$ implies that $r_{G',T}(T')$ is empty in $I/Z$.
% Suppose
% there is a $w-T$ walk in $G\setminus T'$. Noting that $z\in C$ and $T\subseteq C$, when we take the torso to obtain $G'$ this
% $w-T$ will be short-circuited but will be preserved in $G'\setminus T'$, which is a contradiction. Hence there is no $w-T$ walk
% in $G\setminus T'$, i.e., $w\in r_{G,T}(T')$. But $r_{G,T}(T')\subseteq Z$ and therefore we have $w\in Z$, which is again contradiction.
% Thus $r_{G,T}(T')\subseteq Z$ in $G$ implies that $r_{G',T}(T')$ is empty in $I/Z$.
\end{proof}

For every $Z_i$ in the output of Theorem~\ref{thm:main-covering-shadow-derandomized}, we use the torso operation to remove the
vertices in $Z_i$. We prove that this procedure is safe in the following sense:
\begin{lemma}%$[\star]$
%{\bf (correctness of the \shadowless\ algorithm)}
\label{lem:correctalg} Let $I=(G,S,T,k)$ be an instance of \textsc{Disjoint \sdfvs\ Compression}. Let the sets in the output of
Theorem~\ref{thm:main-covering-shadow-derandomized} be $Z_1,Z_2,\ldots,Z_t$. For every $i\in [t]$, let $G_i$ be the reduced
instance $G/Z_i$.
\begin{enumerate}
\item If $I$ is a no-instance, then $G_i$ is also a no-instance for every $i\in [t]$.
\item If $I$ is a yes-instance, then there exists a solution $T^*$ of $I$ which is a shadowless solution of some $G_j$ for
    some $j\in [t]$.
\end{enumerate}
\end{lemma}
\begin{proof}
The first claim is easy to see: any solution $T'$ of the reduced instance $(G_i,S,T,k)$ is also a solution of $(G,S,T,k)$ (by
Lemma~\ref{reduced-instance-soln}(1), the torso operation does not make the problem easier by creating new solutions).

By the derandomization of \textsc{Covering} algorithm, there is a $j\in [t]$ such that $Z$ has the Property $(*)$, i.e., there
is a solution $T^{*}$ of $I$ such that $Z\cap T^{*}=\emptyset$ and $Z$ covers shadow of $T^{*}$. Then
Lemma~\ref{reduced-instance-soln}(2) implies that $T^{*}$ is a shadowless solution for the instance $G_{j}=I/Z_{j}$.
\end{proof}

%%% Local Variables:
%%% mode: latex
%%% TeX-master: "dsfvs-arxiv"
%%% End:

%% file: torso.pdf_tex
%% Creator: Inkscape 0.48.3.1, www.inkscape.org
%% Accompanies image file 'torso.pdf' (pdf, eps, ps)
\begingroup%
  \makeatletter%
  \providecommand\color[2][]{%
    \errmessage{(Inkscape) Color is used for the text in Inkscape, but the package 'color.sty' is not loaded}%
    \renewcommand\color[2][]{}%
  }%
  \providecommand\transparent[1]{%
    \errmessage{(Inkscape) Transparency is used (non-zero) for the text in Inkscape, but the package 'transparent.sty' is not loaded}%
    \renewcommand\transparent[1]{}%
  }%
  \providecommand\rotatebox[2]{#2}%
  \ifx\svgwidth\undefined%
    \setlength{\unitlength}{595.27558594bp}%
    \ifx\svgscale\undefined%
      \relax%
    \else%
      \setlength{\unitlength}{\unitlength * \real{\svgscale}}%
    \fi%
  \else%
    \setlength{\unitlength}{\svgwidth}%
  \fi%
  \global\let\svgwidth\undefined%
  \global\let\svgscale\undefined%
  \makeatother%
  \begin{picture}(1,1.4142857)%
    \put(0,0){\includegraphics[width=\unitlength]{torso.pdf}}%
    \put(0.52969689,0.97773254){\color[rgb]{0,0,0}\makebox(0,0)[b]{\smash{The graph $G$}}}%
    \put(0.18824264,1.11084316){\color[rgb]{0,0,0}\makebox(0,0)[b]{\smash{$v_1$}}}%
    \put(0.38311037,1.11074867){\color[rgb]{0,0,0}\makebox(0,0)[b]{\smash{$v_2$}}}%
    \put(0.64517386,1.11074867){\color[rgb]{0,0,0}\makebox(0,0)[b]{\smash{$v_3$}}}%
    \put(0.83216446,1.11193509){\color[rgb]{0,0,0}\makebox(0,0)[b]{\smash{$v_4$}}}%
    \put(0.29566138,1.29015874){\color[rgb]{0,0,0}\makebox(0,0)[b]{\smash{$v_5$}}}%
    \put(0.73137081,1.29099081){\color[rgb]{0,0,0}\makebox(0,0)[b]{\smash{$v_7$}}}%
    \put(0.4905291,1.29015874){\color[rgb]{0,0,0}\makebox(0,0)[b]{\smash{$v_6$}}}%
    \put(0.2016818,0.80836771){\color[rgb]{0,0,0}\makebox(0,0)[b]{\smash{$v_1$}}}%
    \put(0.4032691,0.82180686){\color[rgb]{0,0,0}\makebox(0,0)[b]{\smash{$v_2$}}}%
    \put(0.63845429,0.80836771){\color[rgb]{0,0,0}\makebox(0,0)[b]{\smash{$v_3$}}}%
    \put(0.82660244,0.81508729){\color[rgb]{0,0,0}\makebox(0,0)[b]{\smash{$v_4$}}}%
    \put(0.55119822,0.66735109){\color[rgb]{0,0,0}\makebox(0,0)[b]{\smash{The graph $\texttt{torso}(G,C,S)$}}}%
  \end{picture}%
\endgroup%

%% file: finding-shadowless.tex
\section{\textsc{Disjoint \sdfvs\ Compression}: Finding a Shadowless Solution}
\label{sec:branch-imp-sep}
%\daniel{``Finding a shadowless solution'' would be a better title}

Consider an instance $(G,S,T,k)$ of \textsc{Disjoint \sdfvs\
  Compression}.  First, let us assume that we can reach a start point
of some edge of $S$ from each vertex of $T$, since otherwise we can
clearly remove such a vertex from the graph (and from the set $T$)
without changing the problem.  Next, we branch on all
$2^{O(k^2)}\log^2 n$ choices for $Z$ taken from
$\{Z_1,Z_2,\ldots,Z_t\}$ (given by
Theorem~\ref{thm:main-covering-shadow-derandomized}) and build a
reduced instance $I/Z$ for each choice of $Z$. By
Lemma~\ref{lem:correctalg}, if $I$ is a no-instance, then
$I/Z_{j}$ is a no-instance for each $j\in [t]$. If $I$ is a
yes-instance, then by Lemma~\ref{lem:correctalg} there is at least one $i\in
[t]$ such that $I$ has a shadowless solution for the reduced instance
$I/Z_{i}$.

Let us consider the branch where $Z=Z_i$ and let $T' \subseteq V \setminus T$ be a hypothetical shadowless solution for $I/Z$.
We know that each vertex in $G\setminus T'$ can reach some vertex of $T$ and can be
reached from a vertex of $T$.
Since $T'$ is a solution for the instance $(G,S,T,k)$ of \textsc{Disjoint \sdfvs\ Compression}, we
know that $G\setminus T'$ does not have any $S$-closed-walks.
Consider a topological ordering $C_{1}$, $C_{2}$, $\ldots$, $C_{\ell}$ of the
strongly connected components of $G\setminus T'$, i.e., there can be an edge from $C_i$ to $C_j$ only if $i\le j$. We
illustrate this in Figure~\ref{fig:strong-components}.

\begin{figure}[t]
\centering
\includegraphics[width=5in]{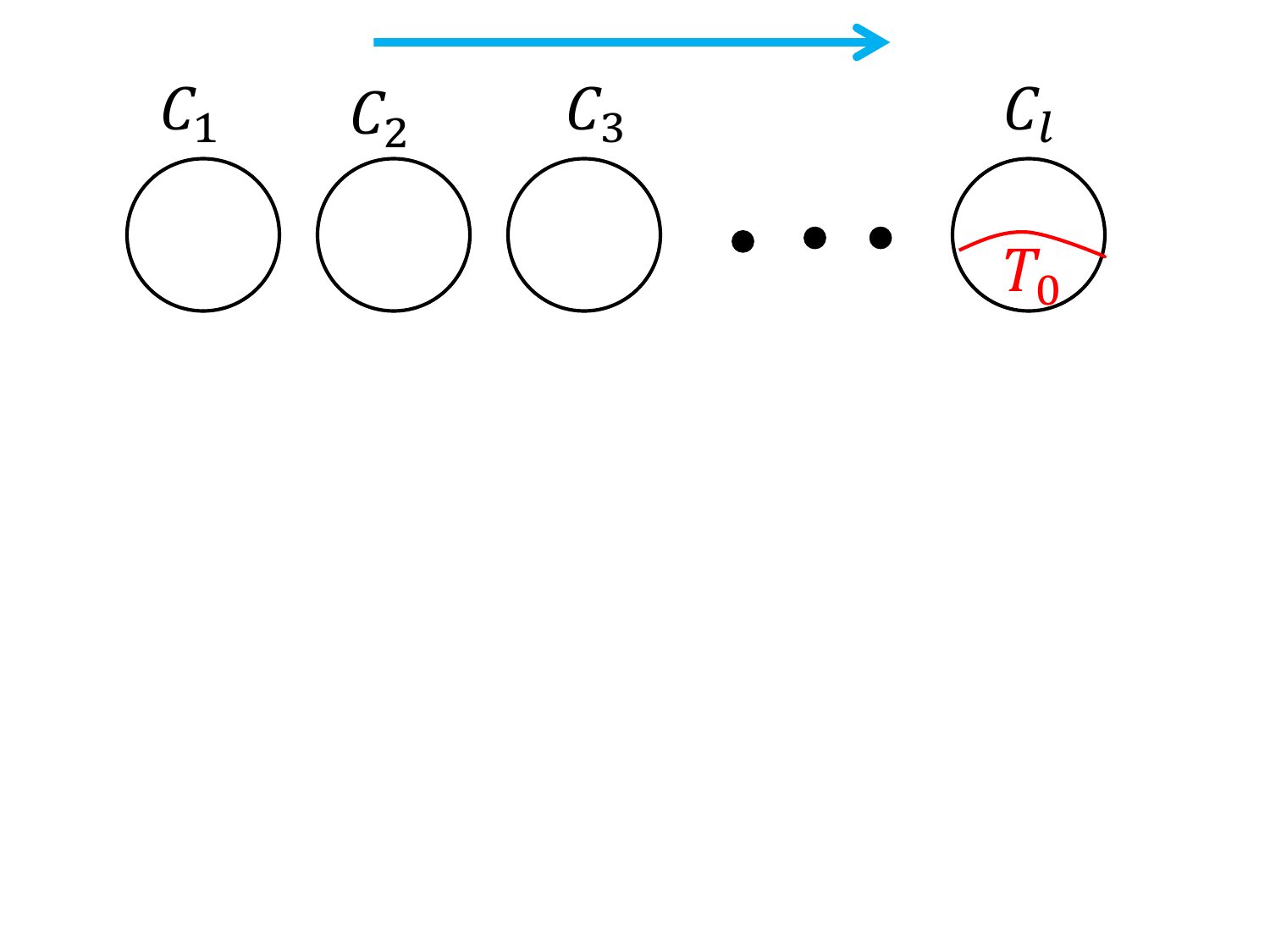}
\vspace{-65mm}
\caption{We arrange the strongly connected components of $G\setminus T'$ in a topological order so that the only possible direction of edges
between the strongly connected components is as shown by the blue arrow. We will show later that the last component $C_{\ell}$ must contain
a non-empty subset $T_0$ of $T$ and further that no edge of $S$ can be present within $C_{\ell}$. This allows us to make
some progress as we shall see in Theorem~\ref{thm:pushing-branch-new}.\label{fig:strong-components}}
\end{figure}

\begin{definition}
{\bf (starting/ending points of $S$)} Let $S^{-}$ and $S^{+}$ be the sets of starting and ending points of edges in $S$ respectively, i.e.,
  $S^{-} = \{u\ |\ (u,v)\in S\}$ and
  $S^{+} = \{v\ |\ (u,v)\in S\}$.
\end{definition}

\begin{lemma}\label{lem:c-ell-prop}
{\bf (properties of $C_{\ell}$)}
For a shadowless solution $T'$ for an instance of \textsc{Disjoint \sdfvs\ Compression},
let $C_{\ell}$ be the last strongly connected component in the topological ordering of $G\setminus T'$
(refer to Figure~\ref{fig:strong-components}). Then
\begin{enumerate}
\item \label{prop:1} $C_{\ell}$ contains a non-empty subset $T_0$ of $T$.
\item No edge of $S$ is present within $C_{\ell}$.
\item \label{prop:3} For each edge $(u,v) \in S$ with $u\in C_{\ell}$, we have $v \in T'$.
%\item $S^{-}$ is disjoint from $C_{\ell}$.
\item \label{prop:4} If $T' \cap S^+=\emptyset$, then $C_{\ell} \cap S^-=\emptyset$.
\end{enumerate}
\label{lem:preliminary-branch-new}
\end{lemma}

\begin{proof}
\begin{enumerate}
\item If $C_{\ell}$ does not contain any vertex from $T$, then the vertices of $C_{\ell}$ cannot reach any vertex of $T$ in $G\setminus T'$. This means that $C_{\ell}$ is in the (reverse) shadow of $T'$, which is a contradiction to the fact that $T'$ is shadowless.
\item If $C_{\ell}$ contains an edge of $S$, then we will have an $S$-closed-walk in the strongly connected component $C_{\ell}$, which
    is a contradiction, as $T'$ is a solution for the instance $(G,S,T,k)$ of \textsc{Disjoint \sdfvs\ Compression}.
\item Consider an edge $(u,v) \in S$ such that $u \in C_{\ell}$ and $v \not\in T'$.
All outgoing edges from $u$ must lie within $C_{\ell}$, since $C_{\ell}$ is the last strongly connected component.
In particular $v\in C_{\ell}$, which contradicts the second claim of the lemma.
\item Assume that $(u,v) \in S$ and $u \in C_{\ell}$ (which means $u \in C_{\ell} \cap S^-$).
      Since $T'$ contains no vertex of $S^+$ we have $v \not\in T'$ and by the third property we have $u \not\in C_{\ell}$, a contradiction.
\end{enumerate}
\end{proof}

Lemma~\ref{lem:c-ell-prop} suggests that we can start by guessing the (nonempty) subset $T_0\subseteq T$ of vertices appearing in the last component $C_\ell$.
Given a set $X$ of removed vertices, we say that edge $(u,v)\in S$ is
{\em traversable} from $T_0$ in $G\setminus X$ if $u,v\not\in X$ and
vertex $u$ (and hence $v$) is reachable from $T_0$ in $G\setminus
X$. If $T'$ is a shadowless solution, then
Lemma~\ref{lem:preliminary-branch-new}(2) implies that no edge of $S$
is traversable from $T_0$ in $G\setminus T'$. There are two ways of
making sure that an edge $(u,v)\in S$ is not traversable: (i) by
making $u$ unreachable from $T_0$, or (ii) by including $v$ in $T'$. The
situation is significantly simpler if every edge of $S$ is handled the
first way, that is, $S^-$ is unreachable from $T_0$ in $G\setminus
T'$. Then $T'$ contains a $T_0-S^-$ separator, and (as we shall see later) we may assume that $T'$ contains an important $T_0-S^-$ separator. Therefore, we can proceed by branching on choosing an important $T_0-S^-$ separator of size at most $k$ and including it into the solution.

The situation is much more complicated if some edges of $S$ are handled the
second way. Given a set $X$ of vertices, we say that an edge $(u,v)\in
S$ is {\em critical} (with respect to $X$) if $v\in X$ and $u$ is
reachable from $T_0$ in $G\setminus X$. Our main observation is that only a bounded number of vertices can be the head of a critical edge in a solution. Moreover, we can enumerate these vertices (more precisely, a bounded-size superset of these vertices) and therefore we can branch on including one of these vertices in the solution.  We describe next how to enumerate these vertices.

% The last property ensures that if no vertex of $S^{+}$ is removed, i.e. $S^{+} \cap T'=\emptyset$,
% then we have an additional property that $S^{-}$ is disjoint from $C_{\ell}$,
% which (as we will show) is enough to branch using important separators.
% However there might be vertices of $S^{+}$ which are part of $T'$
% and our strategy is to find a superset of those vertices
% and branch on which of those vertices to remove (if any).

Let us formalize the property of the vertices we are looking for:
\begin{definition}
\label{def:critical}
{\bf (critical vertex)}
For a fixed non-empty set $T_0 \subseteq V$,
a vertex $v \in (V \setminus T_0) \cap S^{+}$ is called an {\em $\ell$-critical vertex}, with respect
to $T_0$, if there exists an edge $(u,v) \in S$
and a set $W \subseteq V \setminus T_0$ such that:
\begin{itemize}
  \item $|W| \le \ell$,
%  \item $W$ is a $T_0-T \setminus T_0$ separator,
  \item edge $(u,v)$ is critical with respect to $W$ (that is, $u$ is reachable from $T_0$ in $G\setminus W$ and $v\in W$),
  \item no edge of $S$ is traversable from $T_0$ in $G\setminus W$.
\end{itemize}
We say that $v$ is witnessed by $u$, $T_0$ and $W$.
\end{definition}

We need an upper bound on the number of critical vertices,
furthermore our proof needs to be algorithmic, as we want
to find the set of critical vertices, or at least a bounded-size superset of this.
Roughly speaking, to test if $v$ is a critical vertex, we
need to check if there is a set $T'$ that ``cuts away'' every edge of $S$
from $T_0$ in a way that some vertex $u$ with
$(u,v)\in S$ is still reachable from $T_0$. One could argue that it is
sufficient to look at important separators: if there is such a
separator where $u$ is reachable from $T_0$, then certainly there is
an important separator where $u$ is reachable from $T_0$. However,
describing the requirement as ``cutting away every edge of $S$ from $T_0$'' is imprecise:
what we need is that no edge of $S$ is traversable from $T_0$, which
cannot be simply described by the separation of two sets of vertices. We fix this
problem by moving to an auxiliary graph $G'$ by duplicating vertices;
whether or not an edge of $S$ is traversable from $T_0$ translates to
a simple reachability question in $G'$. However, due to technical
issues that arise from this transformation, it is not obvious how to enumerate
precisely the $k$-critical vertices. Instead, we construct a set $F$ of bounded size that
contains each $k$-critical vertex, and potentially some additional vertices.
% This is sufficient for our purposes: if an edge $(u,v)\in S$ is critical in a solution $T'$, then
% $T'$ is a witness of the critical vertex $v$, hence $F$ contains at least
% one vertex of the solution $T'$.
Thus if the solution has a critical
edge, then we can branch on including a vertex of $F$ into the
solution.

% In the following theorem we do not show how to find all the critical vertices,
% however the set of vertices found intersects all the witnesses
% of all the other critical vertices.

\begin{theorem}
\label{thm:critical-enumerate}
{\bf (bounding critical vertices)}
Given a directed graph $G$, a subset $S$ of its edges,
and a fixed non-empty subset $T_0 \subseteq V(G)$,
we can find in time $O^{*}(2^{O(k)})$ a set $F_{T_0}$ of $2^{O(k)}$ vertices
that is a superset of all $k$-critical vertices with respect to $T_0$.
\end{theorem}

\begin{proof}
We create an auxiliary graph $G'$, where the vertex set of $G'$
consists of two copies for each vertex of $V$ and two extra vertices $s$ and $t$, i.e., $V(G') = \{v_\textup{in}, v_\textup{out} : v \in V \}\cup \{s,t\}$. The edges of $G'$ are defined as follows  (see also Fig.~\ref{fig2}):
\begin{itemize}
\item For each edge $e=(u,v) \in E(G)$, we add the following edges to $E(G')$:
if $e \in S$, then add to $E(G')$ an edge $(u_\textup{out},v_\textup{in})$,
otherwise add to $E(G')$ an edge $(u_\textup{out},v_\textup{out})$.
\item For each vertex $v\in V$, we add to $E(G')$ an edge $(v_\textup{in},v_\textup{out})$.
\item For each vertex $v\in V$, we add an edge $(v_\textup{in},t)$ to $E(G')$.
%\item For each vertex $v\in T\setminus T_0$, we add an edge $(v_\textup{out},t)$ to $E(G')$.
\item For each vertex $v\in T_0$, we add an edge $(s,v_\textup{out})$ to $E(G')$.
\end{itemize}

\begin{figure}[t]
\centering
\includegraphics[width=5in]{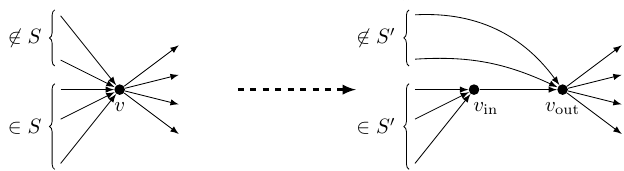}
\caption{On the left there is a vertex $v$ of $G$ and on the right the corresponding vertices $v_\textup{in}$ and $v_\textup{out}$ of $G'$.}
\label{fig2}
\end{figure}

% Define $V'_\textup{in} = \{v_\textup{in} : v \in V\}$,
%  $V'_{F_{T_0}} = \{v_\textup{in}, v_\textup{out} : v
% \in F_{T_0}\}$, $T' = \{v_\textup{in}, v_\textup{out} : v \in T\}$ and as $S'$ take all
% the arcs entering $V'_\textup{in}$ in $G'$.  Moreover, for $X \subseteq V$,
% we denote by $V'_\textup{out}(X)$ the set $\{v_\textup{out} : v \in X\}$.

Let $F_{T_0}'$ be the set of vertices of $G'$ which
belong to some important $s-t$ separator of size at most $2k$.
By Lemma~\ref{number-of-imp-sep} the cardinality of $F_{T_0}'$
is at most $2k\cdot 4^{2k}$.
We define $F_{T_0}$ as $\{v \in V : v_\textup{in} \in F'_{T_0}\}$.
Clearly, the claimed upper bound of $2^{O(k)}$ on $|F_{T_0}|$ follows,
hence it remains to prove that each $k$-critical vertex belongs to $F_{T_0}$.

Let $x$ be an arbitrary $k$-critical vertex
witnessed by $u$, $T_0$ and $W$.
Define $W' = \{v_\textup{in}, v_\textup{out} : v \in W\}$ and note that $|W'| \le 2k$.
The only out-neighbors of $s$ are $\{v_\textup{out}\ |\ v\in T_0\}$ while the only in-neighbors of $t$ are $\{v_\textup{in}\ |\ v\in V\}$. Hence the existence of an $s-t$ path in $G'$ implies that there is in fact
an edge $(a,b)\in S$ that is traversable from $T_0$ in $G\setminus W$ (at some point we have to go from an ``out" vertex to an ``in" vertex, and the only possible way to do this is via an edge from $S$). This is a contradiction to Definition~\ref{def:critical}. Therefore, no
in-neighbor of $t$ is reachable from $s$ in $G'\setminus
W'$, i.e., $W'$ is an $s-t$ separator.  Finally, a path
from $T_0$ to $u$ in $G \setminus W$ translates into a path from
$s$ to $u_\textup{out}$ in $G' \setminus W'$.
Consider an important $s-t$ separator $W''$, i.e., $|W''|\leq |W'|$ and $R^+_{G' \setminus W'}(s) \subset R^+_{G' \setminus W''}(s)$.
As $u_\textup{out}$ is reachable from $s$ in $G' \setminus W'$
we infer that $u_\textup{out}$ is also reachable from $s$ in $G' \setminus W''$.
Consequently $x_\textup{in} \in W''$, as otherwise there would be
an $s-t$ path in $G' \setminus W''$.
Hence $x_\textup{in}$ belongs to $F_{T_0}'$,
which implies that $x$ belongs to $F_{T_0}$ and the theorem follows.
\end{proof}

The following theorem characterizes a solution, so that we can find a vertex contained in it
by inspecting a number of vertices in $V$ bounded by a function of $k$. We apply Theorem~\ref{thm:critical-enumerate} for each subset $T_0\subseteq T$ and let $F = \bigcup_{T_0\subseteq T} F_{T_0}$.
Note that $|F|\leq 2^{|T|}\cdot 2^{O(k)} = 2^{O(|T|+k)}$, and we can generate $F$ in time $2^{|T|}\cdot O^{*}(2^{O(k)}) = O^{*}(2^{O(|T|+k)})$

\begin{theorem}
{\bf (pushing)}
\label{thm:pushing-branch-new}
Let $I=(G,S,T,k)$ be an instance of  \textsc{Disjoint \sdfvs\ Compression} having a shadowless solution
and let $F$ be a set generated by the algorithm of Theorem~\ref{thm:critical-enumerate}. Let $G^+$ be obtained from $G$ by introducing a new vertex $t$ and adding an edge $(u,t)$ for every $u\in S^{-}$.
Then there exists a solution $T' \subseteq V\setminus T$ for $I$ such that either
\begin{itemize}
\item $T'$ contains a vertex of $F \setminus T$, or
\item $T'$ contains an important $T_{0}-(\{t\}\cup (T\setminus T_{0}))$ separator of $G^+$
for some non-empty $T_0 \subseteq T$.
\end{itemize}
\end{theorem}

\begin{proof}
Let $T'$ be any shadowless solution for $I$  and 
let $T_0$ be the subset of $T$ belonging to the last
strongly connected component of $G \setminus T'$; by
Property~\ref{prop:1} of Lemma~\ref{lem:preliminary-branch-new}, $T_0$ is nonempty.

We consider two cases: either there is a $T_0-S^{-}$ path in $G\setminus T'$ or not.
First assume that  there is a path from $T_0$ to a vertex $u \in S^{-}$ in $G \setminus T'$.
Clearly, $u \in C_{\ell}$, since all vertices of $T_0$ belong to $C_{\ell}$
and no edge from $C_{\ell}$ can go to previous strongly connected components.
Consider any edge from $S$ that has $u$ as its starting point, say $(u,v) \in S$.
By Property~\ref{prop:3} of Lemma~\ref{lem:preliminary-branch-new}, we know that $v \in T'$.
Observe that $v$ is a $k$-critical vertex witnessed by $u$, $T_0$, and $T'$, since $|T'| \le k$,
%by definition of $T_0$, there is no path from $T_0$ to $T\setminus T_0$ in $G \setminus T'$;
by definition of $u$, there is a path from $T_0$ to $u$ in $G \setminus T'$; and
by Property~\ref{prop:3} of Lemma~\ref{lem:preliminary-branch-new}, no edge of $S$ is traversable from $T_0$.
Consequently, by the property of the set $F$,
we know that $v \in T' \cap F \neq \emptyset$ and the theorem holds.

Now we assume that no vertex of $S^-$ is reachable from $T_0$ in
$G\setminus T'$.  By the definition of $T_0$, the set $T'$ is a
$T_0-(T\setminus T_0)$ separator in $G$, hence we infer that $T'$ is a
$T_0-(\{t\}\cup (T\setminus T_{0}))$ separator in $G^+$.  Let $T^*$ be the
subset of $T'$ reachable from $T_0$ without going through any other
vertices of $T'$. Then $T^*$ is clearly a $T_{0}-(\{t\}\cup
(T\setminus T_{0}))$ separator in $G^+$. Let $T^{**}$ be the minimal
$T_{0}-(\{t\}\cup (T\setminus T_{0}))$ separator contained in $T^*$.
If $T^{**}$ is an important $T_{0}-(\{t\}\cup (T\setminus T_{0}))$
separator, then we are done, as $T'$ itself contains $T^{**}$.

Otherwise, there exists an important $T_{0}-(\{t\}\cup (T\setminus T_{0}))$ separator $T^{***}$ that dominates $T^{**}$,
i.e., $|T^{***}|\leq |T^{**}|$ and $R^{+}_{G^+\setminus T^{**}}(T_0)\subset R^{+}_{G^+\setminus T^{***}}(T_0)$. Now we claim that
$T'' = (T'\setminus T^{**})\cup T^{***}$ is a solution for the instance $(G,S,T,k)$ of \textsc{Disjoint \sdfvs\ Compression}. If
we show this, then we are done, as $|T''|\leq |T'|$ and $T''$ contains the important $T_{0}-(\{t\}\cup
(T\setminus T_{0}))$ separator $T^{***}$.

Suppose $T''$ is a not a solution for the instance $(G,S,T,k)$ of
\textsc{Disjoint \sdfvs\ Compression}. We have $|T''|\leq |T'|\leq k$
(as , $|T^{***}|\leq |T^{**}|$) and $T''\cap T = \emptyset$ (as $T^{***}$ is an important
$T_{0}-(\{t\}\cup (T\setminus T_{0}))$ separator of $G^+$, hence disjoint from
$T$).  Therefore, the only possible problem is that there is an
$S$-closed-walk in $G\setminus T''$ passing through some vertex $v\in
T^{**}\setminus T^{***}$; in particular, this implies that there is a
$v-S^{-}$ walk in $G\setminus T''$. Since $T^{**}$ is a minimal
$T_{0}-(\{t\}\cup (T\setminus T_{0}))$ separator and $R^{+}_{G^+\setminus T^{**}}(T_0)\subset R^{+}_{G^+\setminus T^{***}}(T_0)$, we have
$(T^{**}\setminus T^{***})\subseteq R^{+}_{G^+\setminus T''}(T_0)$,
implying $v\in R^{+}_{G^+\setminus T''}(T_0)$. This gives a
$T_{0}-S^{-}$ walk via $v$ in $G\setminus T''$, a contradiction as
$T''$ contains an (important) $T_{0}-(\{t\}\cup (T\setminus T_{0}))$
separator by construction.
\end{proof}

Theorem~\ref{thm:pushing-branch-new} tells us that there is always a minimum solution
which either contains some critical vertex of $F$ or an important $T_{0}-(\{t\}\cup (T\setminus T_{0}))$ separator of $G^+$ where $T_0$ is a non-empty subset of $T$.
In the former case, we branch into $|F|$ instances, in each of which we put one vertex of $F$ to the solution,
generating $2^{O(|T|+k)}$ instances with reduced budget.
Next we can assume that the solution does not contain any vertex of $F$ and we try all $2^{|T|}-1$ choices for $T_0$.
For each guess of $T_0$ we enumerate at most $4^{k}$ important $T_{0}-(\{t\}\cup (T\setminus T_{0}))$
separators of size at most $k$ in time $O^{*}(4^k)$ as given by Lemma~\ref{number-of-imp-sep}.
This gives the branching algorithm described in Algorithm~\ref{alg:branch}.

\begin{algorithm}[t]
\caption{\branch \label{alg:branch}}
\textbf{Input:} An instance $I=(G,S,T,k)$ of \textsc{Disjoint \sdfvs\ Compression}. \\
\textbf{Output:} A new set of $2^{O(|T|+k)}$ instances of \textsc{Disjoint \sdfvs\ Compression} where the budget $k$ is reduced.

\begin{algorithmic}[1]

\FOR {every vertex $v \in F \setminus T$ found by Theorem~\ref{thm:critical-enumerate}} \label{line:for-critical}
  \STATE Create a new instance $I_v=(G\setminus v,S,T,k-1)$ of \textsc{Disjoint \sdfvs\ Compression}.
\ENDFOR

\FOR {every non-empty subset $T_0$ of $T$:} \label{line:for}

   \STATE Use Lemma~\ref{number-of-imp-sep} to enumerate all the at most $4^{k}$ important $T_{0}-(\{t\}^{-}\cup (T\setminus
T_{0}))$ separators of size at most $k$ in $G^+$.
   \STATE Let the important separators be $\mathcal{B}=\{B_1,B_2,\ldots,B_m\}$.
   \FOR {each $i\in [m]$}
        \STATE Create a new instance $I_{T_0,i}=(G\setminus B_i,S,T,k-|B_i|)$ of \textsc{Disjoint \sdfvs\ Compression}.
   \ENDFOR
\ENDFOR

\end{algorithmic}
\end{algorithm}

%%% Local Variables:
%%% mode: latex
%%% TeX-master: "dsfvs-arxiv"
%%% End:

%% file: fpt-disjoint.tex
\section{\textsc{Disjoint \sdfvs\ Compression}: Summary of Algorithm}
\label{sec:fpt-algorithm}

Lemma~\ref{lem:correctalg} and the \branch\ algorithm together combine to give a
{bounded search tree} FPT algorithm for \textsc{Disjoint \sdfvs\ Compression} described in Algorithm~\ref{alg:sdfvs}.

\begin{algorithm}[t]
\caption{\textbf{FPT Algorithm for \sdfvs} \label{alg:sdfvs}}
\begin{center}
\noindent{\begin{minipage}{6.00in}
\underline{Step 1}: For a given instance $I=(G,S,T,k)$, use
Theorem~\ref{thm:main-covering-shadow-derandomized} to obtain a set
    of instances $\{Z_1,Z_2,\ldots,Z_t\}$ where $t=2^{O(k^2)}\log^2 n$ and Lemma~\ref{lem:correctalg} implies
    \begin{itemize}
    \item If $I$ is a no-instance, then all the reduced instances $G_{j}=G/Z_j$ are no-instances for all $j\in [t]$
    \item If $I$ is a yes-instance, then there is at least one $i\in [t]$ such that there is a solution $T^{*}$ for $I$
        which is a shadowless solution for the reduced instance $G_{i}=G/Z_i$.
    \end{itemize}
At this step we branch into $2^{O(k^2)}\log^2 n$ directions.\\

\underline{Step 2} : For each of the instances obtained from the above step, we run the \branch\ algorithm to obtain a set of
$2^{O(k+|T|)}$ instances where in each case either the answer is NO, or the budget $k$ is reduced. We solve these instances recursively and return YES if at least one of them returns YES.
\end{minipage}}
\end{center}
\end{algorithm}

We then repeatedly perform Steps 1 and 2. Note that for every instance, one execution of steps 1 and 2 gives rise to
$2^{O(k^2)}\log^2 n$ instances such that for each instance, either we know that the answer is NO or the budget $k$ has
decreased, because we have assumed that from each vertex of $T$ one can reach the set $S^-$, and hence each important
separator is non-empty. Therefore, considering a level as an execution of Step 1 followed by Step 2, the height of the search
tree is at most $k$. Each time we branch into at most $2^{O(k^2)}\log^2 n$ directions (as $|T|$ is at most $k+1$). Hence the
total number of nodes in the search tree is $\Big(2^{O(k^2)}\log^2 n\Big)^{k}$.

\begin{lemma}%$[\star]$
\label{lem:bound-log} For every $n$ and $k\leq n$, we have $(\log n)^k \leq (2k\log k)^{k} + \frac{n}{2^k}$ (the logs are to base 2)
\end{lemma}
\begin{proof}
If $\frac{\log n}{1+\log \log n}\geq k$, then $n\geq (2\log n)^{k}$. Otherwise we have $\frac{\log n}{1+\log \log n}< k$ and
then $(4k\log k)\geq (2\log n)$ as follows: $2k\log k \geq \frac{2\log n\log k}{1+\log \log n}$. Now $ \frac{2\log n\log k}{1+\log \log n} \geq \log n \Leftrightarrow 2\log k \geq 1+\log \log n \Leftrightarrow k^2 \Leftrightarrow 2\log n$. But, $\frac{k^{2}}{2\log n} = \frac{\log n}{2(1+\log\log n)^{2}}$ which is greater than 1 for $n\geq 2^{2^7}$.
\end{proof}

The total number of nodes in the search tree is $\Big(2^{O(k^2)}\log^2 n\Big)^{k} = \Big(2^{O(k^2)}\Big)^{k}(\log^{2} n)^k
= (2^{O(k^3)})(\log^{2} n)^{k} \leq (2^{O(k^3)})\Big( (2k\log k)^{k} + \frac{n}{2^k} \Big)^{2} \leq 2^{O(k^3)}n^{2}$.

We then check the leaf nodes and see if there are any $S$-closed-walks left even after the budget $k$ has become zero. If the
graph in at least one of the leaf nodes is $S$-closed-walk free, then the given instance is a yes-instance. Otherwise it is a
no-instance. This gives an $O^{*}(2^{O(k^3)})$ algorithm for \textsc{Disjoint \sdfvs\ Compression}. By Lemma~\ref{lem:ic}, we
have an $O^{*}(2^{O(k^3)})$ algorithm for the \sdfvs\ problem.

%%% Local Variables:
%%% mode: latex
%%% TeX-master: "dsfvs-arxiv"
%%% End:

%% file: conclusion.tex
\section{Conclusion and Open Problems}
\label{sec:concl-open-probl}

In this paper we gave the first fixed-parameter algorithm for \textsc{Directed Subset Feedback Vertex Set} parameterized by
the size of the solution. Our algorithm used various tools from the FPT world such as iterative compression, bounded-depth
search trees, random sampling of important separators, etc. We also gave a general family of problems for which we can do
random sampling of important separators and obtain a set which is disjoint from a minimum solution and covers its shadow. We
believe this general approach will be useful for deciding the fixed-parameter tractability status of other problems in
directed graphs, where we do not know that many techniques unlike undirected graphs.

The next natural question is whether \sdfvs\ has a polynomial kernel or can we rule out such a possibility under some standard
assumptions? The recent developments~\cite{kernel-3,kernel-1,kernel-2} in the field of kernelization may be useful in
answering this question. In the
field of exact exponential algorithms, Razgon~\cite{razgon-dfvs-exact} gave an $O^{*}(1.9977^{n})$ algorithm for \dfvs\ which was used by Chitnis et al.~\cite{exact-ipec} to give an $O^{*}(1.9993^n)$ algorithm for the more general \sdfvs\ problem. It would be interesting to improve either of this algorithms.
%It would be interesting to break the
%trivial $2^{n}n^{O(1)}$ barrier for \sdfvs.

%% file: dsfvs-arxiv.bbl
\begin{thebibliography}{10}
\providecommand{\url}[1]{\texttt{#1}}
\providecommand{\urlprefix}{URL }

\bibitem{bafna-2-approx-ufvs}
Bafna, V., Berman, P., Fujito, T.: {A 2-Approximation Algorithm for the
  Undirected Feedback Vertex Set Problem}. SIAM J. Discrete Math.  12(3),
  289--297 (1999)

\bibitem{ufvs-2}
Becker, A., Bar-Yehuda, R., Geiger, D.: {Randomized Algorithms for the Loop
  Cutset Problem}. J. Artif. Intell. Res. (JAIR)  12,  219--234 (2000)

\bibitem{ufvs-1}
Bodlaender, H.L.: {On Disjoint Cycles}. In: WG. pp. 230--238 (1991)

\bibitem{fvs-mixed-graphs}
Bonsma, P., Lokshtanov, D.: {Feedback Vertex Set in Mixed Graphs}. In: WADS.
  pp. 122--133 (2011)

\bibitem{ufvs-3}
Cao, Y., Chen, J., Liu, Y.: {On Feedback Vertex Set: New Measure and New
  Structures}. In: SWAT. pp. 93--104 (2010)

\bibitem{ufvs-ic}
Chen, J., Fomin, F.V., Liu, Y., Lu, S., Villanger, Y.: {Improved algorithms for
  feedback vertex set problems}. J. Comput. Syst. Sci.  74(7),  1188--1198
  (2008)

\bibitem{chen-improved-multiway-cut}
Chen, J., Liu, Y., Lu, S.: {An Improved Parameterized Algorithm for the Minimum
  Node Multiway Cut Problem}. Algorithmica  55(1),  1--13 (2009)

\bibitem{chen-dfvs}
Chen, J., Liu, Y., Lu, S., O'Sullivan, B., Razgon, I.: {A fixed-parameter
  algorithm for the directed feedback vertex set problem}. J. ACM  55(5) (2008)

\bibitem{DBLP:conf/icalp/ChitnisCHM12}
Chitnis, R.H., Cygan, M., Hajiaghayi, M.T., Marx, D.: Directed subset feedback
  vertex set is fixed-parameter tractable. In: ICALP (1). pp. 230--241 (2012)

\bibitem{exact-ipec}
Chitnis, R.H., Fomin, F.V., Lokshtanov, D., Misra, P., Ramanujan, M.S.,
  Saurabh, S.: Faster exact algorithms for some terminal set problems. In:
  {IPEC}. pp. 150--162 (2013)

\bibitem{directed-multiway-cut}
Chitnis, R.H., Hajiaghayi, M., Marx, D.: {Fixed-Parameter Tractability of
  Directed Multiway Cut Parameterized by the Size of the Cutset}. {SIAM} J.
  Comput.  42(4),  1674--1696 (2013)

\bibitem{kernel-3}
Cygan, M., Kratsch, S., Pilipczuk, M., Pilipczuk, M., Wahlstr{\"{o}}m, M.:
  {Clique Cover and Graph Separation: New Incompressibility Results}. {TOCT}
  6(2), ~6 (2014)

\bibitem{cut-and-count}
Cygan, M., Nederlof, J., Pilipczuk, M., Pilipczuk, M., van Rooij, J.M.M.,
  Wojtaszczyk, J.O.: {Solving Connectivity Problems Parameterized by Treewidth
  in Single Exponential Time}. In: FOCS. pp. 150--159 (2011)

\bibitem{DBLP:journals/toct/CyganPPW13}
Cygan, M., Pilipczuk, M., Pilipczuk, M., Wojtaszczyk, J.O.: On multiway cut
  parameterized above lower bounds. TOCT  5(1), ~3 (2013)

\bibitem{usfvs-icalp}
Cygan, M., Pilipczuk, M., Pilipczuk, M., Wojtaszczyk, J.O.: Subset feedback
  vertex set is fixed-parameter tractable. SIAM J. Discrete Math.  27(1),
  290--309 (2013)

\bibitem{mikefellows-fvs}
Dehne, F., Fellows, M.R., Langston, M.A., Rosamond, F.A., Stevens, K.: {An
  O(2$^{\mbox{O(k)}}$n$^{\mbox{3}}$) FPT Algorithm for the Undirected Feedback
  Vertex Set Problem}. Theory Comput. Syst.  41(3),  479--492 (2007)

\bibitem{df-fvs}
Downey, R.G., Fellows, M.R.: {Fixed-Parameter Tractability and Completeness I:
  Basic Results}. SIAM J. Comput.  24(4),  873--921 (1995)

\bibitem{downey-fellows}
Downey, R.G., Fellows, M.R.: {Parameterized Complexity}. Springer-Verlag
  (1999), 530 pp.

\bibitem{even-dfvs-approx}
Even, G., Naor, J., Schieber, B., Sudan, M.: {Approximating Minimum Feedback
  Sets and Multi-Cuts in Directed Graphs}. In: IPCO. pp. 14--28 (1995)

\bibitem{even-8-approx-sufvs}
Even, G., Naor, J., Zosin, L.: {An 8-Approximation Algorithm for the Subset
  Feedback Vertex Set Problem}. SIAM J. Comput.  30(4),  1231--1252 (2000)

\bibitem{flum-grohe}
Flum, J., Grohe, M.: {Parameterized Complexity Theory}. Springer-Verlag (2006),
  493 pp.

\bibitem{saket-ic}
Fomin, F.V., Gaspers, S., Kratsch, D., Liedloff, M., Saurabh, S.: {Iterative
  compression and exact algorithms}. Theor. Comput. Sci.  411(7-9),  1045--1053
  (2010)

\bibitem{wernicke-fvs}
Guo, J., Gramm, J., H{\"u}ffner, F., Niedermeier, R., Wernicke, S.:
  {Compression-based fixed-parameter algorithms for feedback vertex set and
  edge bipartization}. J. Comput. Syst. Sci.  72(8),  1386--1396 (2006)

\bibitem{huffner-ic}
H{\"u}ffner, F., Komusiewicz, C., Moser, H., Niedermeier, R.: {Fixed-Parameter
  Algorithms for Cluster Vertex Deletion}. In: LATIN. pp. 711--722 (2008)

\bibitem{usfvs-ken}
Kakimura, N., Kawarabayashi, K., Kobayashi, Y.: Erd{\"o}s-p{\'o}sa property and
  its algorithmic applications: parity constraints, subset feedback set, and
  subset packing. In: SODA. pp. 1726--1736 (2012)

\bibitem{ufvs-7}
Kanj, I.A., Pelsmajer, M.J., Schaefer, M.: {Parameterized Algorithms for
  Feedback Vertex Set}. In: IWPEC. pp. 235--247 (2004)

\bibitem{karp-np-hardness}
Karp, R.M.: {Reducibility Among Combinatorial Problems}. In: Complexity of
  Computer Computations. pp. 85--103 (1972)

\bibitem{DBLP:conf/icalp/KratschPPW12}
Kratsch, S., Pilipczuk, M., Pilipczuk, M., Wahlstr{\"o}m, M.: Fixed-parameter
  tractability of multicut in directed acyclic graphs. In: ICALP (1). pp.
  581--593 (2012)

\bibitem{kernel-2}
Kratsch, S., Wahlstr{\"o}m, M.: {Compression via matroids: a randomized
  polynomial kernel for odd cycle transversal}. In: SODA. pp. 94--103 (2012)

\bibitem{kernel-1}
Kratsch, S., Wahlstr{\"{o}}m, M.: {Representative Sets and Irrelevant Vertices:
  New Tools for Kernelization}. In: {FOCS}. pp. 450--459 (2012)

\bibitem{lokshtanov-marx-clustering}
Lokshtanov, D., Marx, D.: Clustering with local restrictions. Inf. Comput.
  222,  278--292 (2013)

\bibitem{lp-daniel}
Lokshtanov, D., Narayanaswamy, N.S., Raman, V., Ramanujan, M.S., Saurabh, S.:
  Faster parameterized algorithms using linear programming. {ACM} Transactions
  on Algorithms  11(2), ~15 (2014)

\bibitem{DBLP:conf/icalp/LokshtanovR12}
Lokshtanov, D., Ramanujan, M.S.: {Parameterized Tractability of Multiway Cut
  with Parity Constraints}. In: ICALP (1). pp. 750--761 (2012)

\bibitem{marx-2006}
Marx, D.: {Parameterized graph separation problems}. Theor. Comput. Sci.
  351(3),  394--406 (2006)

\bibitem{marx-razgon-stoc-11}
Marx, D., Razgon, I.: Fixed-parameter tractability of multicut parameterized by
  the size of the cutset. {SIAM} J. Comput.  43(2),  355--388 (2014)

\bibitem{mehlhorn}
Mehlhorn, K.: {Data Structures and Algorithms 2: Graph Algorithms and
  NP-completeness}. Springer (1984)

\bibitem{aravind-focs-1995}
Naor, M., Schulman, L.J., Srinivasan, A.: {Splitters and Near-Optimal
  Derandomization}. In: FOCS. pp. 182--191 (1995)

\bibitem{niedermeier}
Niedermeier, R.: {Invitation to Fixed-Parameter Algorithms}. Oxford University
  Press (2006)

\bibitem{saket-fvs-1}
Raman, V., Saurabh, S., Subramanian, C.R.: {Faster Fixed Parameter Tractable
  Algorithms for Undirected Feedback Vertex Set}. In: ISAAC. pp. 241--248
  (2002)

\bibitem{saket-fvs-2}
Raman, V., Saurabh, S., Subramanian, C.R.: {Faster fixed parameter tractable
  algorithms for finding feedback vertex sets}. ACM Transactions on Algorithms
  2(3),  403--415 (2006)

\bibitem{razgon-dfvs-exact}
Razgon, I.: {Computing Minimum Directed Feedback Vertex Set in
  O(1.9977$^{\mbox{n}}$)}. In: ICTCS. pp. 70--81 (2007)

\bibitem{almost-2-sat}
Razgon, I., O'Sullivan, B.: {Almost 2-SAT is fixed-parameter tractable}. J.
  Comput. Syst. Sci.  75(8),  435--450 (2009)

\bibitem{reed-smith-vetta-ic}
Reed, B.A., Smith, K., Vetta, A.: {Finding odd cycle transversals}. Oper. Res.
  Lett.  32(4),  299--301 (2004)

\bibitem{seymour-dfvs-approx}
Seymour, P.D.: {Packing Directed Circuits Fractionally}. Combinatorica  15(2),
  281--288 (1995)

\end{thebibliography}
